%% file: main.tex
\documentclass[11pt]{article}
\usepackage{bruce-common}
\usepackage{mary-additional}

\newcommand\gtuppereq{d^2 \log(m/d)}
\newcommand\gtupper{O(\gtuppereq)}

\definecolor{nicepurple}{RGB}{106,27,154}
\definecolor{nicepink}{RGB}{232,82,133}
\definecolor{nicepeach}{RGB}{242,164,159}

\title{Unconstraining graph-constrained group testing
}
\author{Bruce Spang\thanks{Department of Computer Science, Stanford University.  \href{mailto:bspang@cs.stanford.edu}{bspang@cs.stanford.edu}} \and Mary Wootters\thanks{Departments of Computer Science and Electrical Engineering, Stanford University.  \href{mailto:marykw@stanford.edu}{marykw@stanford.edu}}}
\begin{document}

\maketitle

\begin{abstract}
In  network tomography, one goal is to identify a small set of failed links in a network, by sending a few packets through the network and seeing which reach their destination.
This problem can be seen as a variant of combinatorial group testing, which has been studied before under the moniker \em graph-constrained group testing. \em

The main contribution of this work is to show that for most graphs, the ``constraints'' imposed by the underlying network topology are no constraint at all. That is, the number of tests required to identify the failed links in ``graph-constrained'' group testing is near-optimal even for the corresponding group testing problem \em with no graph constraints. \em
Our approach is based on a simple randomized construction of tests; to analyze our construction, we prove new results about the size of giant components in randomly sparsified graphs.

Finally, we provide empirical results which suggest that our connected-subgraph tests perform better not just in theory but also in practice, and in particular perform better on a real-world network topology.

\end{abstract}

\section{Introduction}

\input{intro}

\section{Setup and Preliminaries}\label{sec:prelim}
\input{prelim}

\section{Related Work}\label{sec:related}
\input{related}
\section{Our Results}\label{sec:results}

\input{results}
\section{Proofs}\label{sec:proofs}
\input{expansion}
\section{Empirical Results}\label{sec:empirical}
\input{empirical}

\section{Conclusion}\label{sec:conclusion}
\input{conclusion}

\section*{Acknowledgements}
We thank Cl\'{e}ment Canonne and Nick McKeown for helpful comments.

\bibliographystyle{alpha}
\bibliography{existing.bib}

\appendix
\section{The power of connected-subgraph tests.}\label{app:separation}
\input{separation}

\end{document}

%% file: intro.tex
Suppose you run a network with $n$ switches and $m$ links between the switches.  Occasionally links will fail, and it is your goal to find and fix them.  In practice, even finding a failure can be non-trivial: often the only available clue is the failure of some traffic which interacts with many links.
In light of this, the problem of \em network tomography \em (see \cite{Castro:2004fj} for a survey) is: what can be learned about a network by observing only its traffic?

In the version of the problem we will study here, suppose that some links---at most $d$ of them---fail.
We may try to send a packet along any connected path through the network, and we observe whether or not the packet reaches its destination. The goal is to identify any set of up to $d$ failed links while sending as few packets as possible.  We focus on the \em non-adaptive \em setting, where the paths for the packets must be fixed ahead of time.  Non-adaptive tests are faster since they allow packets to be sent in parallel, and are easier to implement since these paths will be hard-coded into the switches.

As observed by~\cite{Harvey:2007ez, Cheraghchi:2010vm}, this problem is a variant of a well-studied problem called \em combinatorial group testing. \em Combinatorial group testing, originally motivated by the problem of cheaply testing for disease~\cite{D43}, has been studied since the 1940's and has applications from computational biology to wireless networks.  We refer the reader to \cite{Du:1999by} for a survey.
In the combinatorial group testing problem there are $m$ items, at most $d$ of which are ``defective.'' 
A single test of a subset $T$ of items reveals whether or not there are any defective items in $T$.
The goal is to identify the defective items, by observing the output of a few tests.

The connection to network tomography is as follows: each link is defective if it fails, and each test $T$ corresponds to a set of links. In the network tomography setting, there is one additional requirement: a test $T \subseteq [m]$ must correspond to a path that a packet could take through the network.  Because of this connection, \cite{Cheraghchi:2010vm} called this problem ``graph-constrained group testing.''

\paragraph{Our Question.}
A natural question is whether the additional constraints of graph-constrained group testing necessitate significantly more tests.  For the unconstrained group testing problem, the optimal number of tests is essentially $\gtupper$.\footnote{The optimal number of tests for unconstrained group testing has recently been shown to be $\Theta( d^2 \log_d(m))$: a breakthrough work of D'yachkov et al.~\cite{D14} proved the upper bound, which settled a question that had been open since the lower bound was proved in the early 1980's~\cite{DR82}.  However, we focus on the benchmark $\gtupper$, because it is optimal up to a factor of $\log d$, easy to see that a random set of tests achieves this in the unconstrained setting, and even this benchmark was not previously known in the graph-constrained setting.
} 
Thus, our question is as follows:
\begin{question}\label{q:thequestion}
For what graphs $G$ can we solve the graph-constrained group testing problem using $\gtupper$ tests?
\end{question}

Previous work \cite{Harvey:2007ez} has shown that certain graphs, such as a line, require $\Omega(m)$ tests, far more than the $\gtupper$ we would need for the unconstrained problem. However, 
it is also known that for sufficiently ``well-connected'' graphs 
(for example, those with large minimum cuts, many disjoint spanning trees, or small mixing time), a sublinear number of tests suffice~\cite{Harvey:2007ez,Cheraghchi:2010vm}.  These works have proposed using large subtrees~\cite{Harvey:2007ez} or random walks~\cite{Cheraghchi:2010vm} as the tests.  However, both of these approaches stop short (by polylogarithmic factors or more) of obtaining an $\gtupper$ bound.

\paragraph{Our contributions.}
We improve upon the results of \cite{Harvey:2007ez,Cheraghchi:2010vm} to show that $\gtupper$ tests are sufficient for a wide collection of graphs, including many of the graphs already considered in prior work.
Our construction---which is randomized---is quite simple: we sparsify the graph by choosing edges at random, and use the resulting large connected subgraphs.  
This is similar in flavor to earlier work---for example, \cite{Cheraghchi:2010vm} considered random walks---but our tests lead to stronger theorems, and also appear to perform better in practice. 
Concretely, our contributions are as follows:
\begin{itemize}
\item \textbf{$\gtupper$ tests suffice for $(\beta, \alpha)$-edge expanders.}
Our main result, Theorem~\ref{thm:mainedge}, applies to graphs which are $(\beta, \alpha)$-edge expanders, meaning that every set $S \subseteq V$ of size at most $\beta n$ has at least $\alpha|S|$ edges coming out of it.  We show that $\gtupper$ tests suffice when $\beta$ is constant and $\alpha \gtrsim d$.
Moreover, if $\beta$ is sub-constant, then the number of tests required degrades gracefully with $\beta$.

$(\beta,\alpha)$-edge expansion is a general notion, and our results imply improved graph-constrained group-testing schemes
for several natural classes of graphs like Erd\"{o}s-R\'{e}nyi graphs and constant-degree expanders.  Moreover, our results are even optimal when applied to certain ``counter-example'' graphs like the barbell graph which foil earlier work. 

Our general theorem (Theorem~\ref{thm:mainedge}) is compared to existing general theorems in Table~\ref{table:num-tests}.  The results for a few specific families of graphs are shown in Table~\ref{tab:specific}.  These results are presented in more detail in Section~\ref{sec:results}.

\begin{table}
\renewcommand{\arraystretch}{1.4}
  \centering
  \begin{tabular}{|c|c|c|c|}
    \hline
    \textbf{Source} & \textbf{Graph} & \textbf{Max. defective edges}                       & \textbf{Number of Tests}                          \\ \hline\hline
    \hspace{1sp}\cite{Harvey:2007ez} & $G$ has min-cut $K$ & $d \leq \ceil{\frac{K-1}{2}}-1$ & $O(d^3\log_d(m))$                        \\
	\hline
    \hspace{1sp}\cite{Cheraghchi:2010vm} & \begin{minipage}{4cm}\begin{center} \vspace{.3cm} $G$ is $D$-regular, \\ with mixing time $\tau$ \vspace{.2cm} \end{center} \end{minipage} & \begin{minipage}{3cm}\begin{center} \vspace{.3cm} $d \leq d_0$ for some \\ $d_0 = \Omega(D/\tau^2)$ \vspace{.2cm} \end{center} \end{minipage} & $O(\tau^2\gtuppereq)$                     \\ 
	\hline
    Proposition~\ref{thm:cutedge} & $G$ has min-cut $K$ & $d \leq \frac{K}{5\log n}$       & $\gtupper$             \\
	\hline
    Theorem~\ref{thm:mainedge} & $G$ is a $(\beta,\alpha)$-edge expander & $d \leq \inparen{ \frac{1}{2} + 0.01 }\alpha$       & $O\left(\frac{\gtuppereq}{\beta}\right)$              \\ \hline
  \end{tabular}
  \caption{Summary of general results using connected-subgraph tests to identify any $d$ defective edges for ``well-connected'' graphs with $n$ vertices and $m$ edges, for various notions of ``well-connected.''
See discussion in Section~\ref{sec:related} for slightly more general statements of results in previous work.  } 
  \label{table:num-tests}
\end{table}
\begin{table}
\renewcommand{\arraystretch}{1.4}
\centering
\begin{tabular}{|p{4cm}|c|c|c|}
\hline
\textbf{Graph} & \textbf{Source} & \textbf{Number of tests required} & \begin{minipage}{4cm} \begin{center} \vspace{.5cm} \textbf{Limit $d_0$ so that recovery of $d \leq d_0$ failures is possible}  \vspace{.5cm} \end{center}\end{minipage}\\\hline\hline
  \nameref{par:complete}
      & \cite{Cheraghchi:2010vm} & \cellcolor{green!15} $\gtupper$ &\cellcolor{green!15}  $d_0 = \Omega(n)$ \\\cline{2-4}
      & This work & \cellcolor{green!15} $\gtupper$ &\cellcolor{green!15}  $d_0 =  \Omega(n)$ \\ \hline\hline
  \nameref{par:expanders}
               & \cite{Harvey:2007ez} & $O( d^3 \log_d(m) )$ &\cellcolor{green!15}  $d_0 = \Omega(D)$ \\\cline{2-4}
      & \cite{Cheraghchi:2010vm} & $O( d^2 \log^3(m) )$ & $d_0 = \Omega(D/\log^2(n))$  \\\cline{2-4}
      & This work & \cellcolor{green!15} $\gtupper$ &\cellcolor{green!15}  $d_0 = \Omega(D)$ \\\hline\hline
  \nameref{par:erdos-renyi} $G(n,D/n)$
      & \cite{Cheraghchi:2010vm} & $O( d^2 \log^3(m) )$ & $d_0 = \Omega(D/\log^2(n))$ \\\cline{2-4}
      & This work & \cellcolor{green!15} $\gtupper$ &\cellcolor{green!15}  $d_0 = \Omega(D)$ \\\hline\hline
  \nameref{par:barbells}
      & \cite{Harvey:2007ez} & $O(d^3 \log_d m)$ & $d_0 = 1$ \\\cline{2-4}
      & \cite{Cheraghchi:2010vm} & $m$ \hyperref[par:barbells]{(see discussion)} & \cellcolor{green!15} $d_0 = n$ \\\cline{2-4}
      & This work & \cellcolor{green!15} $\gtupper$ & \cellcolor{green!15} $d_0 = \Omega(n)$ \\\hline
\end{tabular}
\caption{Summary of work on the number of connected-subgraph tests required to identify any $d \leq d_0$ failures, for specific families of graphs.  All graphs have $n$ vertices and $m$ edges. Results which meet the near-optimal $\gtupper$ tests or which have only the asymptotically optimal restriction $d \leq d_0$ for some $d_0 = \Omega(\text{degree})$ are highlighted in green.}
\label{tab:specific}
\end{table}

\item \textbf{New results about large connected components of random graphs.}
While our construction is quite simple, the analysis requires some delicacy.  In order to show that our tests work, we prove new results about giant components in randomly sparsified graphs.

More precisely, 
our main technical theorem (Theorem~\ref{thm:giantcomponent}) establishes the following.  Suppose that $G = (V,E)$ is a $(\beta, \alpha)$-edge expander, and let $G(p) = G(V,E')$ be the graph where $E' \subseteq E$ is a random subset where each edge is kept independently with probability $p$.  We show that if $p \geq (1 + \eps)/\alpha$, then for any edge $e \in E$, with probability $\Omega(p\eps)$ then not only does $e$ survive, but also $e$'s connected component in $G(p)$ has size at least $\beta n$.  This is of a similar flavor to previous work on giant components of randomly sparsified graphs, but with two important differences: first, our result works even if $\beta$ is small (so the components are ``big'' but not ``giant'') and second, we require that any edge $e$ be contained in a large component with decent probability after sparsification.  Theorem~\ref{thm:giantcomponent} is stated and proved in Section~\ref{sec:proofs}.

\item 
  Finally, we present empirical results which suggest that our approach significantly out-performs the random-walk method of \cite{Cheraghchi:2010vm} and in many cases, nearly matches the performance of unconstrained random tests. On complete graphs and hypercubes it uses less than half the tests of the random-walk method. On a family of graphs often used in datacenter networks (the ``fat-tree'' topology), our approach is able to find defectives using a nontrivial number of tests while the random-walk method is not.
\end{itemize}

\paragraph{Organization.}
In Section~\ref{sec:prelim}, we formally set up the problem.  In Section~\ref{sec:related}, we survey related work, and we state our theoretical results in Section~\ref{sec:results}.  The proofs of these results follow in Section~\ref{sec:proofs}.  Finally, we present our empirical results in Section~\ref{sec:empirical}.

%% file: prelim.tex
We begin with some basic notation and definitions. 

\paragraph{Graph-theoretic preliminaries.}
Throughout, we will be working with undirected, unweighted graphs $G = (V,E)$ with $|V| = n$, $|E| = m$.
For a set of vertices $A \subseteq V$, the \em boundary \em of $A$ is
\[ \partial A = \inset { \{u,v\} \in E \suchthat u \in A, v \not\in A }. \]
For a set of edges $B \subseteq E$, we use the notation $N(B)$ to denote the set of vertices $v$ that are endpoints of an edge in $B$:
\[ N(B) = \inset{ v \in V \suchthat \exists u, \{u,v\} \in B }. \]

The minimum cut $K$ of a graph $G$ is defined by $K = \min_{A \subseteq V} |\partial A|.$
Our main theorem is about \em edge expanders. \em  We give a slightly more general definition than the usual notion (which would have $\beta = 1/2$ below), so that we can state a more general theorem.
\begin{definition}
A graph $G = (V,E)$ is a $(\beta, \alpha)$-edge expander if for all sets $A \subseteq V$ with $|A| \leq \beta |V|$, $|\partial A| \geq \alpha |A|$.
\end{definition}

We will consider random sparsifications of graphs.  For a graph $G = (V,E)$ and $p \in (0,1)$, $G(p) = (V,E')$ denotes the random graph where $E' \subseteq E$ is generated by including each edge of $E$ in $E'$ independently with probability $p$.  We use $G(n,p) = K_n(p)$ to denote the Erd\"{o}s-R\'{e}nyi graph where each edge is included independently with probability $p$.  (Here, $K_n$ denotes the complete graph on $n$ vertices).

\paragraph{Group testing preliminaries.}
The combinatorial group testing problem is set up as follows (using slightly non-standard notation in order to be consistent with the graph-constrained set-up below).  Let $E$ be a universe of size $m$, and suppose that $B \subseteq E$ is a set of at most $d$ special or ``defective'' items in $E$.
A \em test \em $T \subseteq E$ is a collection of items, and we say that the \em outcome \em of the test $T$ is \textsc{True} if $T \cap B \neq \emptyset$ and \textsc{False}  otherwise.
We say that a collection of tests $\mathcal{T} \subseteq 2^E$
(here, $2^E$ denotes the power set of $E$, consisting of all subsets of $E$) 
\em can identify up to $d$ defective items in $E$ \em  if for any $B \subseteq E$ with $|B| \leq d$, $B$ is uniquely determined from the outcomes of the tests $T \in \mathcal{T}$.  The goal is to design a collection of tests $\mathcal{T} \subseteq 2^E$ which can identify up to $d$ defective items, so that $|\mathcal{T}|$ is as small as possible.
A useful notion in the group testing literature is \em disjunctness, \em  which is a sufficient condition for recovery.
\begin{definition}
Let $E$ be a universe and $\mathcal{T} \subseteq 2^E$.  We say that $\mathcal{T}$ is \em $d$-disjunct \em if for all $e \in E$, for all $B \subseteq E$ where $|B| \leq d$ and $e \not\in B$, there exists a test $T \in \mathcal{T}$ so that $e \in T$ and $B \cap T = \emptyset$.
\end{definition}
If $\mathcal{T}$ is $d$-disjunct, then $\mathcal{T}$ can identify up to $d$ defective items in $E$.  More precisely, it is not hard to see that the following algorithm will do the job: for each item $e \in E$, declare $e \in B$ if and only if all the tests $T \in \mathcal{T}$ with $e \in T$ had outcome \textsc{True}. 

Choosing tests completely at random is a good way to obtain $d$-disjunct sets.
\begin{prop}[See, e.g., \hspace{1sp}\cite{Du:1999by} Theorem 8.1.3] \label{prop:gt-upper}
Let $d \geq 1$.
Let $E$ be a universe. Consider a random test $T \subseteq E$ such that each $e \in E$ is included in $T$ independently with probability $p = \frac{1}{d+1}$.
Let $\mathcal{T} = \{T_1, \ldots, T_\tau\}$, where each $T_i \in \mathcal{T}$ is chosen independently from the above distribution.  Then there is a value $\tau = \gtupper$ so that $\mathcal{T}$ is $d$-disjunct with probability at least $1 - 1/m$.
\end{prop}

This is nearly optimal, up to a factor of $O(1/\log d)$:
\begin{theorem}[\cite{DR82}] \label{thm:gt-lower}
  Let $E$ be a universe of size $m$ and $\mathcal{T} \subseteq 2^E$. If $\mathcal{T}$ is $d$-disjunct, then $|\mathcal{T}| = \Omega(d^2\log_d m)$
\end{theorem}

\paragraph{Graph-constrained group testing.}
Given a graph $G = (V,E)$,
the graph-constrained group testing problem on $G$ is the same as the standard group testing problem on a universe of items $E$, with the additional constraint that each test $T \subseteq E$ be a connected set.  That is, in the network tomography application, a test $T$ must be able to be traversable by a packet.
We say that a \em connected-subgraph test \em is a set of edges $T \subseteq E$ so that $(N(T), T)$ is a connected set.  

We note that the definition of disjunctness directly applies to the graph-constrained setting, and our goal in this work will be to design $d$-disjunct collections $\mathcal{T}$ of connected-subgraph tests, such that $\mathcal{T}$ is as small as possible.

\begin{remark}[Why connected-subgraph tests?]
Our work, like existing work on graph-constrained group testing \cite{Harvey:2007ez,Cheraghchi:2010vm,Karbasi:2012gd}, uses connected-subgraph tests. Connected-subgraph tests can be implemented in a programmable network \cite{Bosshart:2014dr} by hard-coding routing for test packets, or by using source routing.
One could also imagine restricting the tests to be, for example, simple paths or trees.
Some restrictions on tests do have some advantages in implementation (in particular, tests which are \em shortest \em paths may be easier to implement than general connected-subgraph tests: for example many networks support Equal-cost multi-path routing (ECMP) which splits traffic across all the shortest paths between a pair of hosts).
However, connected subgraph tests are strictly more powerful than either simple paths or trees for the constrained group-testing problem.  
We provide some examples which demonstrate this in Appendix~\ref{app:separation}.
\end{remark}

\begin{remark}[Why worst-case failures?]\label{remark:random}
We focus on the worst-case failure model as it is the most conservative and lines up with previous work on graph-constrained group-testing.  However, we remark that if the $d$ failures occur uniformly at random, our proof still goes through and shows that the performance of our connected-subgraph tests are roughly the same as performance of unconstrained random tests.  This would imply that our approach can identify $d$ random failures with high probability on $(\beta, \alpha)$-edge expanders, provided that $\alpha \gtrsim d$.
\end{remark}

%% file: related.tex
\paragraph{Boolean network tomography.} Most of the work on \em boolean network tomography \em (that is, the problem of identifying failures in a graph using end-to-end traffic) has a much harser set-up than the one we consider here, in that both the graph and the tests are taken to be worst-case, or at least very constrained. For example, all the tests may be required to be simple paths starting from a particular vertex.  The reason for the harsh set-up is that historically, networks have been quite inflexible, which severely limits both the graph topologies and the sorts of tests that are used.    
Since identifying failures uniquely is often impossible in these settings, this work has focused on doing as well as possible given the circumstances, for example by finding \em any \em set of failures that will explain the test outcomes \cite{Bejerano:2003io, Dhamdhere:2007he} or by finding the most likely set of failures given some underlying distribution~\cite{Duffield:2006ue,Nguyen:2007ds}.
When the input graph and set of allowed tests is worst-case, these problems are hard, and the usual approach is to reduce to some NP-hard problem and use a heuristic or approximation algorithm.

\paragraph{Graph-constrained group testing.}
More recently, the field of networking has shifted towards flexible datacenter networks, where the tomography problem is still interesting \cite{Zeng:2013ff, Roy:2017um}. Modern datacenter networks, however, fundamentally change the constraints of the tomography problem. Datacenters are good expanders \cite{Dinitz:2017wb, Valadarsky:2016wo}, so the worst-case assumptions about the graphs can be relaxed. Modern networks are programmable \cite{Bosshart:2014dr}: instead of the network defining what can and cannot be done, operators program networks to do what they want. Thus, the set of tests $\mathcal{T}$ need not be worst-case.
This leads to graph-constrained group testing, where we can design the tests, and make assumptions about the connectivity of the underlying network.

We are not the first to investigate group testing with graph constraints. 
Du and Hwang discuss two different group-testing problems on graphs in Chapter 12 of \cite{Du:1999by}, although neither are exactly the same as the setup we consider here.
The connection between boolean network tomography and group testing was first observed by~\cite{Harvey:2007ez}.
They give results for specific families of graphs including line graphs, grids, and binary trees. Their most general result is that if a graph has $d$ edge-disjoint spanning trees $T_1,\ldots,T_d$ with $\delta$ being the maximum diameter of the trees, then $O(d^3\log_d(m) + d\min(\delta + \log^2n, \delta\log n))$ tests are sufficient to identify at most $d$ failures.  (As a corollary, this implies that any graph with minimum cut $K$ can identify $d \leq \ceil{ \frac{K-1}{2} } - 1$ failed edges, which is what is stated in Table~\ref{table:num-tests}).

The work closest to our is that of Cheraghchi et al.~\cite{Cheraghchi:2010vm}, who give a randomized construction of connected-subgraph tests via random walks. They show that $\gtupper$ tests suffice for \em very \em well-connected graphs (those with constant mixing time), and $O(d^2\log^3(m))$ tests suffice for certain expanders and Erd\"os-R\'eyni graphs. Their most general result is that for graphs with mixing time $\tau$ and where there exists some $c > 0$ such that the degree $D_v$ of each vertex $v \in V$ lies in between $6c^2d\tau^2 \leq D_v \leq 6c^3d\tau^2$, at most $O(c^4\tau^2\gtuppereq)$ tests are sufficient. 

Finally, \cite{Karbasi:2012gd} considers the adaptive version of graph-constrained group testing.  They present an adaptive algorithm which, on any graph, uses a number of tests that is within a constant factor of the optimal number.

We summarize the general results for related work in non-adaptive graph-constrained group testing in Table~\ref{table:num-tests}.

\paragraph{Giant components in random graphs.}
Finally, we mention some related work on the size of giant components of randomly sparsified graphs, since our main technical theorem (Theorem~\ref{thm:giantcomponent}) is related to this.
This question is well-studied, but we need a slightly different result. 
One difference is that we work with $(\beta, \alpha)$-edge expansion, and in particular our ``giant'' components need not be so giant if $\beta$ is small.  A second difference is that we require that every edge be contained in a large connected component with constant probability; to the best of our knowledge, existing work does not explicitly give such a guarantee.

The study of giant components in $G(n,p)$ (aka, a randomly sparsified complete graph) was initiated by Erd\"{o}s and R\'{e}nyi in \cite{Erdos:1959ud}.
This was extended to sparsifications of the hypercube~\cite{Ajtai:1982jx} and sufficiently good expander graphs~\cite{Frieze:2004kg} and \cite{Chung:2006cw}.   Our approach to Theorem~\ref{thm:giantcomponent} is based on that of Krivelevich and Sudakov~\cite{Krivelevich:2013tn} who give a simpler argument for existing results on giant components.

Most of these results show that as long as $p \leq \frac{ 1 + \eps }{D}$, where $D$ is the degree of the graph $G$, then $G(p)$ contains a giant component.  We show a similar result: if $p \leq \frac{ 1 + \eps }{ \alpha }$, where $G$ is a $(\beta, \alpha)$-edge expander, then there exists a connected component of size at least $\beta n$.  (And moreover, any edge is contained in such a component with decent probability).

%% file: results.tex
In this section, we give a brief overview of our results and approach.

\subsection{Main result}

Our group testing scheme is quite simple: the idea is just to choose random edges of the graph, and keep any large-enough connected components.
Recall the notation that for a graph $G =(E,V)$, $G(p) = (V,E')$ is the graph where each edge in $E$ is included in $E'$ independently with probability $p$.  Then the randomized algorithm for constructing the tests is given in Algorithm~\ref{alg:make-tests}.
\begin{algorithm}
  \SetKwInOut{Input}{input}\SetKwInOut{Output}{output}
  \caption{Make-Tests}\label{alg:make-tests}
  \Input{Graph $G = (V,E)$; number of failed edges $d$; parameters $\delta \geq \frac{2}{d}$, $\beta \in (0, \frac{1}{2}]$, and $\tau \in \N$}
  \Output{A collection of tests $\mathcal{T} \subseteq 2^E$}

	$\mathcal{T} \gets \emptyset$\;
	$p \gets \frac{1}{\delta d}$\;
	\For{ $t = 1, \ldots, \tau$ }
	{
		Draw $G' \sim G(p)$ independently from all the other rounds\;
		Find the connected components $A_1, A_2, \ldots, A_r$ of $G(p)$\;
		\For{ each $A_i$ so that $|A_i| \geq \beta n$ }{
			Add the test $A_i$ to $\mathcal{T}$\;
		}
	}
	\Return $\mathcal{T}$

\end{algorithm}

Our main theorem implies that any $(\beta, \alpha)$-edge expander with large enough $\alpha$ admits a group testing scheme with $O(\gtuppereq/\beta)$ tests.
\begin{theorem}\label{thm:mainedge}
There are constants $c,C > 0$ so that the following holds.
Suppose that $G = (V,E)$ is a graph with $|V| = n, |E| = m$.  Let $d \geq 1$ be an integer, and $\delta \geq \frac{2}{d}$.\footnote{Note that $\delta$ may be larger than $1$ if desired.} 
Suppose that $G$ is a $(\beta, \alpha)$-edge-expander with $\beta \in (0,1/2]$ and such that $\alpha \geq 1$ satisfies
\begin{equation} 
	\alpha \geq d \inparen{ \frac{1}{2} + \delta(1 + \delta) }.  \label{eq:alpha-assmp}
\end{equation}

Let $\mathcal{T}$ be the set of tests returned by Algorithm~\ref{alg:make-tests} run with parameters $\delta$, $\beta$, and $\tau = C \gtuppereq$.  Then with probability at least $1 - m^{-cd}$, $\mathcal{T}$ is $d$-disjunct.  Further,
\[ |\mathcal{T}| \leq  \frac{ C \gtuppereq e^{1/\delta} }{\beta}. \]
\end{theorem}

\begin{remark}[Optimality]
If $\beta, \delta$ are constant, then the number of tests required is $\gtupper$, which is nearly optimal even for the unconstrained group testing problem (that is, it nearly matches Theorem~\ref{thm:gt-lower}).  
Moreover, the requirement on the expansion factor $\alpha$ is nearly tight.  That is, in \eqref{eq:alpha-assmp}, we may take $\alpha = \inparen{ \frac{1}{2} + \gamma} d$ for any constant $\gamma > 0$, while still maintaining the near-optimal $\gtupper$ test complexity.
On the other hand, such a statement could not hold for $\alpha$ much smaller than $d/2$.  More precisely, we show below that the degree $D$ of the graph $G$ must be at least $d/2$ to obtain any nontrivial bound on $|\mathcal{T}|$.  Since for any $\gamma > 0$, there exist $D$-regular $(\beta, \alpha)$-edge expanders with $\alpha = (1 - \gamma)D$ for small $\beta \sim \gamma$, this implies that the cut-off $d/2$ above cannot be improved.
\end{remark}

\begin{prop}\label{thm:cut-lower-bound}
Let $G = (V,E)$ be a $D$-regular graph with $|E| = m$.
If $\mathcal{T} \subseteq 2^E$ is a collection of $d$-disjunct connected-subgraph tests, for $d \geq 2D - 2$,
then $|\mathcal{T}| \geq m$.
\end{prop}
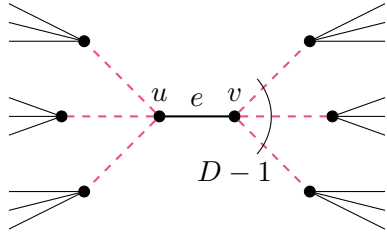
\begin{figure}
\centering
\begin{tikzpicture}
\node[draw,circle, scale=0.4, fill=black](a) at (0,0) {};
\node[draw,circle, scale=0.4, fill=black](b) at (1,0) {};
\node[above=0cm of a] {$u$};
\node[above=0cm of b] {$v$};
\draw[black,thick] (a) to node[above] {$e$} (b);
\node[draw,circle, scale=0.4, fill=black](c) at (-1,1) {};
\node[draw,circle, scale=0.4, fill=black](d) at (-1.3,0) {};
\node[draw,circle, scale=0.4, fill=black](e) at (-1,-1) {};
\draw[thick,nicepink,dashed] (a) to (c);
\draw[thick,nicepink,dashed] (a) to (d);
\draw[thick,nicepink,dashed] (a) to (e);
\node[draw,circle, scale=0.4, fill=black](f) at (2,1) {};
\node[draw,circle, scale=0.4, fill=black](g) at (2.3,0) {};
\node[draw,circle, scale=0.4, fill=black](h) at (2,-1) {};
\draw[thick,nicepink,dashed] (b) to (f);
\draw[thick,nicepink,dashed] (b) to (g);
\draw[thick,nicepink,dashed] (b) to (h);
\draw (c) to (-2, 1.5);
\draw (c) to (-2, 1.25);
\draw (c) to (-2, 1);
\draw (d) to (-2, .25);
\draw (d) to (-2, 0);
\draw (d) to (-2, -.25);
\draw (e) to (-2, -1.5);
\draw (e) to (-2, -1.25);
\draw (e) to (-2, -1);
\draw (f) to (3, 1.5);
\draw (f) to (3, 1.25);
\draw (f) to (3, 1);
\draw (g) to (3, .25);
\draw (g) to (3, 0);
\draw (g) to (3, -.25);
\draw (h) to (3, -1.5);
\draw (h) to (3, -1.25);
\draw (h) to (3, -1);
\draw (1.3, .5) to [out=-50,in=50] (1.3,-.5);
\node at (1, -.75) {$D-1$};
\end{tikzpicture}
\caption{Proof of Proposition~\ref{thm:cut-lower-bound}.  Suppose that the number of edges that may fail is $d \geq 2D - 2$ where $D$ is the degree of the graph.  If the set of tests $\mathcal{T}$ does not contain the singleton $\{e\}$ for an edge $e = \{u,v\}$, then the set $B$ of dashed edges---all of the neighbors of $e$---provides an counter-example to $d$-disjunctness.}\label{fig:pflb}
\end{figure}
\begin{proof}
Suppose that $\mathcal{T} \subseteq 2^E$ with $|\mathcal{T}| < m$.  Then there is some edge $e = \{u,v\}$ so $\{e\} \not\in \mathcal{T}$.  
Let $B = \partial( \{u,v\} )$ be the set of edges adjacent to $e$, so $|B| = 2D - 2 \leq d$.  Then the only connected-subgraph test $T \subset E$ so that $e \in T$ but $T \cap B = \emptyset$ is $\{e\}$, which by assumption is not in $\mathcal{T}$.  Thus, $\mathcal{T}$ is not $d$-disjunct.  (See Figure~\ref{fig:pflb}).
\end{proof}

\subsection{Instantiations for particular graphs}
In this section, we instantiate Theorem~\ref{thm:mainedge} for several families of graphs and compare them to existing results.
We remark that some of these families ($D$-regular expanders, or $G(n,p)$) are natural candidates, while others (like a barbell graph) are concocted to show the difference between our theorem and that of previous work.
A summary is shown in Table~\ref{tab:specific}, and we go into the details below.

\paragraph{Complete Graphs.}\label{par:complete}
The mixing time for a complete graph is constant, so 
\cite{Cheraghchi:2010vm} 
gives an optimal construction requiring $\gtupper$ tests. Theorem~\ref{thm:mainedge} gives the same result.

\paragraph{$D$-Regular Expander Graphs with Constant Spectral Gap.}\label{par:expanders}
$D$-regular expander graphs are $D$-regular graphs which are very ``well-connected.''  One way of measuring this is the \em spectral gap, \em that is the difference between the largest eigenvalue $D$ of the adjacency matrix $A_G$ and the second-largest eigenvalue, $\lambda$.  (We refer the reader to the excellent survey~\cite{Hoory:2006kj} for more background on expander graphs.)
We consider families of $D$-regular graphs $G$ whose second largest eigenvalue $\lambda$ is bounded away from $D$ by a constant: $\lambda \leq D(1 - c)$ for some constant $c \in (0,1)$ independent of $n$. %

For larger $D = \Omega(\log^2 n)$, \cite{Cheraghchi:2010vm} show that $O(d^2\log^3(m))$ tests are sufficient, which is optimal up to logarithmic factors. For smaller $D$, in particular when $D$ is a constant, the best previously known result 
guarantees $O(d^3\log(m/d))$ tests~\cite{Harvey:2007ez}.

As we will see below, Theorem~\ref{thm:mainedge} guarantees $\gtupper$ tests are sufficient for all $D$.
In order to apply Theorem~\ref{thm:mainedge}, we relate the second largest eigenvalue to edge-expansion as follows:

\begin{theorem}[\hspace{1sp}See \cite{Hoory:2006kj} Theorem 4.11] \label{fact:cheegers}
  Let $G=(V,E)$ be a finite, connected, $D$-regular graph and let $\lambda$ be its second eigenvalue. Then
$G$ is $(1/2, \alpha)$-edge expander with
 \[ \frac{D-\lambda}{2} \leq \alpha \leq \sqrt{2D(D-\lambda)}.\]
\end{theorem}

By plugging in Theorem~\ref{fact:cheegers} to Theorem~\ref{thm:mainedge} (with constant and sufficiently small $\delta > 0$) we obtain the following corollary:
\begin{corollary}\label{thm:expanders}
  Let $G$ be a $D$-regular expander with second largest eigenvalue $\lambda \leq D(1-c)$ for some constant $c > 0.$ 
Then for any $d < cD/2$, there is a collection $\mathcal{T}$ of connected subgraph tests so that $\mathcal{T}$ is $d$-disjunct and $|\mathcal{T}| = \gtupper$.
\end{corollary}

Notice that by Proposition~\ref{thm:cut-lower-bound}, the restriction that $d \leq d_0$ for some $d_0 = O(D)$ is necessary.

\paragraph{Erd\"{o}s-R{\'e}nyi Graphs.}\label{par:erdos-renyi}
Like the $D$-regular expanders above,
an Erd\"{o}s-R{\'e}nyi random graph $G(n,p)$ on $n$ nodes with parameter $0 \leq p \leq 1$ is well-connected, and has good spectral properties with high probability.  However, these graphs are not $D$-regular so we consider them separately. 
For larger $p = \Omega( \log^2 n / n)$, \cite{Cheraghchi:2010vm} show that $O(d^2\log^3(m))$ tests are sufficient to guarantee $d$-disjunctness.
Theorem~\ref{thm:mainedge} can improve this to $\gtupper$, with only the restriction $p \geq p_0$ for some $p_0 = \Omega(n)$.

In order to apply Theorem~\ref{thm:mainedge}, we use the following lemma from~\cite{Cheraghchi:2010vm}, which implies that $G(n,D/n)$ for $D = \Omega(d \log n)$ has $(1/2,\alpha)$ edge expansion for $\alpha \geq (2+\epsilon)d$. Theorem~\ref{thm:mainedge} immediately implies that $\gtupper$ tests are sufficient.

\begin{lemma}[\hspace{1sp}\cite{Cheraghchi:2010vm} Lemma 32]
  For every $\phi < 1/2$ there is an $\alpha > 0$ such that a random graph $G = G(n,p)$ with $p \geq \alpha \ln n/n$ has edge expansion $\alpha \geq \phi D$ with probability $1-o(1)$.
\end{lemma}

\paragraph{Barbells.}\label{par:barbells}

One of the advantages of our result over previous work is that the notion of $(\beta, \alpha)$-edge expansion captures a more general notion of ``well-connected'' than is captured by minimum cuts or mixing times.  As an extreme example of this, consider a barbell graph $G$, 
which we define as 
two copies of the complete graph $K_{n/2}$ on $n/2$ vertices, connected by one edge $e$ (Figure~\ref{fig:barbell}).
\begin{wrapfigure}{r}{4cm}
\centering
\begin{tikzpicture}[scale=.7]
\draw[fill=blue!10] (0,0) circle (1cm) ;
\draw[fill=blue!10] (3,0) circle (1cm) ;
\node[draw,circle,fill=black, scale=0.4] (a) at (1,0) {};
\node[draw,circle,fill=black,scale=0.4] (b) at (2,0) {};
\draw (a) to node[above] {$e$} (b);
\node at (0,0) {$K_{n/2}$};
\node at (3,0) {$K_{n/2}$};
\end{tikzpicture}
\caption{A barbell graph}\label{fig:barbell}
\end{wrapfigure}
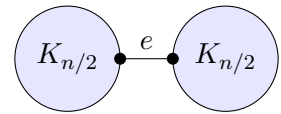

This graph is great for graph-constrained group testing: we test the connecting edge $e$ on its own, then then use $\gtupper$ tests to identify up to $d$ failures in each of the two copies of $K_{n/2}$.  Thus, $\gtupper$ tests are sufficient.
However, this is a worst-case graph for existing work. It has a minimum cut of one edge, so \cite{Harvey:2007ez} only allows $d=1$. The mixing time of this graph is quite large: the probability of reaching the center edge is $O(1/n)$, so certainly $\tau = \Omega(n).$ The degree condition of \cite{Cheraghchi:2010vm} is not satisfied as $D \leq n/2$ which is not $O(1/n^2)$. However, even if we could ignore this condition, \cite{Cheraghchi:2010vm} would use at least $\twiddle{O}(n^2\gtuppereq)$ tests (or, therefore $O(m)$ tests since this is the naive solution). %

On the other hand, our results using $(\beta, \alpha)$-expansion match the intuition that this example should be easy.
Setting $\beta = \frac{1}{4}$ and $\alpha \geq \frac{n}{2} - \frac{n}{4} = \frac{n}{4}$, Theorem~\ref{thm:mainedge} gives an $\gtupper$ bound.

This is example is meant to highlight the difference between our work and existing work.  While the barbell graph is unlikely to be used in practice, it illustrates the intuition that $(\beta, \alpha)$-connectivity does better capture somewhat ``clustery'' graphs---that is, graphs with higher connectivity in some areas than in others---than either minimum cuts or mixing time.   This notion may be useful for real-life networks; for example, networks may have higher connectivity within a rack than between racks.

\subsection{Overview of approach}
Our approach is quite simple: we just choose random edges and take the large connected components.  Intuitively, the reason that this works is because:
\begin{enumerate}
\item If we just chose random edges, we would be back in the traditional group testing setting, where random tests are nearly optimal by Proposition~\ref{prop:gt-upper}.
	\item We will show that---provided $G$ is a $(\beta, \alpha)$-edge expander---then the graph $G(p)$ formed by choosing random edges has mostly large components, of size at least $\beta n$.  Thus, intuitively, throwing away the few disconnected parts should not matter much.
\end{enumerate}

\noindent 
There are several challenges in making the above intuition rigorous:

First, once we throw away the edges that are not in a large connected component, the edges that remain, conditional on remaining, are no longer independent.  Thus, the intuition from point 1 above does not quite hold. However, this ends up being reasonably straightforward to deal with.

The second and more interesting challenge is that
we must show that each edge $e$ is still contained in a test with high probability.  That is, when we pass from $G$ to $G(p) = (V,E')$, the probability that $e \in E'$ is $p$; for the analysis in the random case to still work, we need the probability that $e$ is in a large connected component of $E'$ to also be proportional to $p$.

This second challenge is where the edge expansion of the underlying graph comes in.  To prove that $e$'s connected component in $G(p)$ has a decent probability of being large, we reduce the question to one about random walks. Taking inspiration from \cite{Krivelevich:2013tn}, we introduce a process to generate the connected component of a particular edge $v$, and argue that this process generates a large connected component if and only if an appropriately chosen random walk diverges with decent probability.  Then we prove that this walk indeed diverges.

The details of the approach are given in Section~\ref{sec:proofs}.  However, as a warm-up to show why the intuition presented above is believable, we first prove an easier statement where neither of these challenges arise.

\begin{prop}\label{thm:cutedge}
There is a constant $C > 0$ so that the following holds.
Suppose that $G = (V,E)$ is a graph with $|V| = n$, $|E| = m$.  Let $d \geq 1$ be an integer.
Suppose that the minimum cut $K$ of $G$ satisfies $K \geq 5(d+1)\log(n)$.
Let $\mathcal{T}$ be a set of tests $T \subseteq E$ generated according to the following process:
\begin{itemize}
	\item Initialize $\mathcal{T} = \emptyset$.
	\item For $t = 1, 2, \ldots, C \gtuppereq$:
	\begin{itemize}
		\item Let $T \subset E$ be a random set where each edge is included with probability $1/(d+1)$.
		\item Add $T$ to $\mathcal{T}$.
	\end{itemize}
\end{itemize}
Then $\mathcal{T}$ is $d$-disjunct and each test is connected with high probability.
\end{prop}

\begin{proof}
  Because of Proposition~\ref{prop:gt-upper}, it suffices to show that a random set $T$ is connected with high probability.
Fortunately, this is true:
\begin{theorem}[\hspace{1sp}\cite{Karger:1994wa}]
\label{fact:cut-connected-threshold}
Let $K$ be the minimum cut of $G$. If $p > \min\(\frac{5\log n}{K}, 1\)$, $G(p)$ is connected with probability at least $1-\frac{1}{n}$.
\end{theorem}

By Theorem~\ref{fact:cut-connected-threshold} and a union bound, the probability any test is disconnected is at most $|\mathcal{T}|/n = O\(\frac{\gtuppereq}{n}\)$.
\end{proof}

%% file: expansion.tex
In this section, we prove Theorem~\ref{thm:mainedge}.
Our proof is based on the following theorem, which implies that any edge $e$ is reasonably likely to be contained in a large connected component of $G(p)$.
\begin{theorem}\label{thm:giantcomponent}
Let $\beta \in (0,1/2)$ and $\alpha \geq 1$, and
let $G = (V,E)$ be a graph with $|V| = n, |E| = m$ %
so that,
for any set $A \subset V$ of size $2 \leq |A| \leq \beta n$, we have
\[ |\partial A| \geq \alpha |A|. \]
For an edge $e \in E$, let $C_e$ denote the connected component of $G(p)$ containing $e$, or $\emptyset$ if there is no such connected component (that is, if $e \not\in G(p)$ was deleted), and let $|C_e|$ denote the number of vertices in $C_e$.

Choose any $\eps \in (0,1/3)$, and suppose that
\[ p \geq \frac{ 1 + \eps }{\alpha}. \] %
Then
for all edges $e \in E$,
\begin{equation}
	\PR{ |C_e| \geq \beta n } \geq \frac{ p \eps }{8}.\label{eq:included-wcp}
\end{equation}
\end{theorem}
Before we prove Theorem~\ref{thm:giantcomponent}, we show how it can be used to prove Theorem~\ref{thm:mainedge}.
\begin{proof}[Proof of Theorem~\ref{thm:mainedge}, assuming Theorem~\ref{thm:giantcomponent}]

Let $G = (V,E)$, and let $G(p) = (V, E')$ be the random sparsification.
Fix $B \subseteq E$ with $|B| = d$ and $e \in E$.  We will show that with high probability, at least one of the tests in $\mathcal{T}$ will separate $e$ from $B$.  Then we will union bound over all choices for $e$ and $B$ to conclude that $\mathcal{T}$ is $d$-disjunct.

Consider one draw of $G(p)$ in Make-Tests, and
let $T_1, T_2, \ldots$ be the connected components of $G(p)$ which have size at least $\beta n $.
(That is, $T_1, T_2,\ldots$ are the tests that we add to $\mathcal{T}$).
Then we have
\begin{align*}
\PR{ \exists i \text{ s.t. } e \in T_i \text{ and } B \cap T_i = \emptyset }
&= \PR{ |C_e| \geq \beta n  \text{ and } B \cap C_e = \emptyset } \\
&\leq \PR{ |C_e| \geq \beta n  \text{ and } B \cap E' = \emptyset } \\
&= \PR{ |C_e| \geq \beta n  \,\mid \, B \cap E' = \emptyset } \cdot \PR{ B \cap E' = \emptyset }.
\end{align*}
We have
\[ \PR{ B \cap E' = \emptyset } = (1 - p)^d, \]
since this is just the probability that all the edges in $B$ survive in $G(p)$.

For our fixed $e,B$, let $\bar{G} = (V, E \setminus B)$ be the graph with all the edges in $B$ removed.  
Consider the distribution of $G(p)$ conditioned on the event that $B \cap E' = \emptyset$.  This is the same as the distribution of $\bar{G}(p)$.
To see this, notice that for $A \subseteq E \setminus B$, the random sets of $A \cap E'$ and $B\cap E'$ are independent.
Let $\overline{C_e}$ be the connected component containing $e$ in $\bar{G}(p)$.
Then
\[ \PR{ |C_e| \geq \beta n  \,\mid \, B \cap E' = \emptyset } = \PR{ |\overline{C_e}| \geq \beta n }. \]
Notice that since $G$ is a $(\beta, \alpha)$-edge expander with 
$ \alpha > \frac{d}{2}(1 + \eps)$, then 
for any set $A \subset V$ of size $2 \leq |A| \leq \beta n$, we have
\[ |\partial A| \geq \alpha |A| - d \geq \inparen{\alpha - \frac{d}{2}}|A|. \]
Thus, we may apply Theorem~\ref{thm:giantcomponent} to $\bar{G}$.  
Choose
\[ p = \frac{ 2 }{ d \delta }, \]
and apply Theorem~\ref{thm:giantcomponent} with $\eps = 1$.
Our assumption \eqref{eq:alpha-assmp} that $\alpha \geq d\inparen{\frac{1}{2} + \delta}$ and the choice of $p$ 
implies that
\[ p = \frac{2 }{ d\delta } \geq \frac{ 1 + \eps }{ \alpha - d/2 },  \]
we conclude that
\[ \PR{ |\overline{C_e}| \geq \beta n  }  \geq \frac{p \eps}{8}. \]
Putting things together,
for a single draw of $G(p)$, we have
\[ \PR{ \exists i \text{ s.t. } e \in T_i \text{ and } B \cap T_i = \emptyset }  \geq \frac{\eps p}{2} \cdot (1 - p)^d. \]
Then, by independence, the probability that this happens at least once over $\tau$ draws is at least
\[ \PR{ \exists T \in \mathcal{T} \text{ s.t. } e \in T \text{ and } B \cap T = \emptyset }
\geq 1 - \inparen{ 1 - \frac{ \eps p (1 - p)^d }{8 } }^\tau. \]
By a union bound, the probability that $\mathcal{T}$ is $d$-disjunct is at least
\begin{align*}
\PR{ \mathcal{T} \text{ is $d$-disjunct } } &\geq 1 - m \cdot {m \choose d} \inparen{ 1 - \frac{\eps p(1 - p)^d }{8} }^\tau \\
&\geq 1 - m {m \choose d} \exp \inparen{ \frac{ - \tau \eps p (1 - p)^d }{ 8 } } \\
&\geq 1 - m {m \choose d} \exp \inparen{ \frac{ - \tau \eps p e^{-1/\delta}}{ 16 } } \\
&\geq 1 - \exp \inparen{ (d+1) \ln(m) - \frac{ \tau \eps  p e^{-1/\delta} }{ 16 } }\\
&= 1 - \exp \inparen{ (d+1) \ln(m) - \frac{ \tau \eps   e^{-1/\delta} }{ 16 d \delta } }.
\end{align*}
Above, we have used the choice of $\delta$ to imply that  $p \leq 1/2$, which implies that 
\[(1 - p)^d \geq \frac{1}{2} \exp(-pd) = \frac{1}{2} \exp(-1/\delta). \]
Thus, there are some constants $c,C > 0$ so that if
\[ \tau \geq \frac{C d^2 \delta \exp(1/\delta) \ln(m)}{\eps}, \]
then
we have
\[ \PR{ \mathcal{T} \text{ is $d$-disjunct}} \geq 1 - \exp( - c d \ln(m) ) \]
as desired.

Finally, we observe that the number of tests is at most $\tau / \beta$, since there are at most $1/\beta$ connected components of size $\beta n$ in a graph with $n$ vertices, so there are at most $1 / \beta$ tests $A_i$ added to $\mathcal{T}$ for each $t \in \{1, \ldots, \tau\}$.  
Choosing $\eps = \delta$ completes the proof.
\end{proof}

Next, we prove Theorem~\ref{thm:giantcomponent}.

\begin{proof}[Proof of Theorem \ref{thm:giantcomponent}]
As in the theorem statement, suppose $p \geq (1 + \eps)/\alpha$ for some $\eps \in (0,1/3)$, and
let $G = (V,E)$ be a $(\beta, \alpha)$-edge expander.
Write $G(p) = (V, E')$, so that $E' \subseteq E$.  Choose $p \geq (1 + \eps)/\alpha$.

Consider the probability $(u,v)$ is in a large component,
\[ \PR{C_{(u,v)} \geq \beta n } = \PR{ (u,v) \in E'  } \cdot \PR{ |C_{(u,v)}| \geq \beta n  \,\mid\, (u,v) \in E' }. \]
Condition on the event that $(u,v) \in E'$, and
imagine building $C_{(u,v)}$ by starting with $\{(u,v)\}$ and building the set outwards, one at a time.
More precisely, consider the following randomized process:

\begin{itemize}
\item $S_0 \gets \{(u,v)\}$, $B_0 \gets \emptyset$
\item For $t=0,1,2,\ldots$:
\begin{enumerate}
	\item Let $N(S_t)$ be the set of vertices in $V$ adjacent to $S_t$.
	\item Let $U_t = \partial(N(S_t)) \setminus B_t$ be the set of unvisited edges that lie on the boundary of $N(S_t)$.
	\item If $|U_t| = 0$, \textbf{break}.
	\item $S_{t+1} \gets S_t$, $B_{t+1} \gets B_t$.
	\item Choose an edge $e \in U_t$ arbitrarily.
	\item With probability $p$, declare that $e$ has survived and add it to $S_{t+1}$.
	\item Otherwise (with probability $1-p$) add $e$ to $B_{t+1}$.
\end{enumerate}
\end{itemize}
This process is illustrated in Figure~\ref{fig:search-process}.

  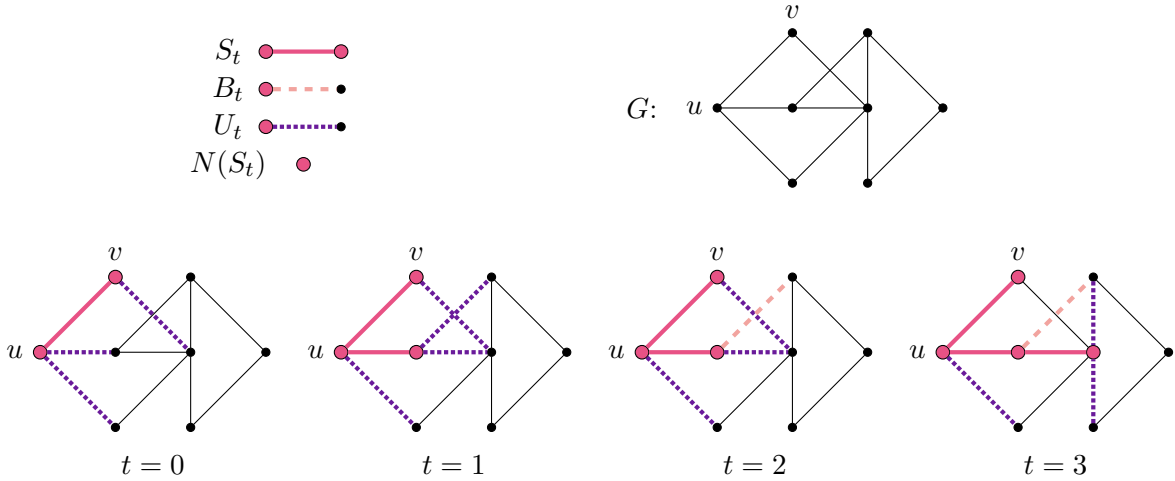
\begin{figure}[H]
    \centering
\input{figure}
    \caption{First few steps of the process to build $C_{(u,v)}$ in the proof of Theorem \ref{thm:giantcomponent}.  The edge from $U_t$ which we chose ended up being in $E'$ in steps $t=1$ and $t=3$, but not $t=2$. }
    \label{fig:search-process}
  \end{figure}

It is not hard to see that this process terminates at the first time $t_{\max}$ so that $|U_{t_{\max}}| = 0$, and when it does, $N(S_{t_{\max}})$ is distributed identically to the set of vertices in $C_{(u,v)}$.
Thus, to bound $|C_{(u,v)}|$ with high probability, we can bound the set $|N(S_t)|$ with high probability.
Notice that $S_t$ is a tree; thus,
\[ |S_t| = |N(S_t)| - 1, \]
and so it suffices to show that
\[ |S_{t_{\max}}| > \beta n +1  \]
with high probability.
To that end, we will show that, as long as $|S_t| \leq \beta n +1,$ the probability that $|U_t| = 0$ is very small.

At each step $t$, we either add an edge to $B_t$ or to $S_t$, so $|S_t| + |B_t| = t$.  Let $X_t$ be the random variable which is $1$ if we added an edge to $S_t$ in step $t$; thus $X_t \sim \text{Ber}(p)$ and $|S_t| = 1 + \sum_{i=1}^t X_i$.

Suppose that $S_t$ is nonempty and $|S_t| \leq \beta n + 1$, so $2 \leq |N(S_t)| \leq \beta n$.  By our expansion assumption, we have
\[ |\partial N(S_t)| \geq \alpha |S_t|, \]
and so
\begin{align*}
 |U_t| &=
| \partial(N(S_t)) \setminus B_t | \\
& \geq \alpha |S_t| - |B_t| \\
& = \alpha |S_t| - (t- |S_t|) \\
&= (1 + \alpha) |S_t| - t \\
&= (1 + \alpha)\inparen{ 1 + \sum_{i=1}^t X_i } - t \\
&\geq \sum_{i=1}^t ((1 + \alpha) X_i - 1) + \alpha.
\end{align*}

Thus,
\begin{align*}
\PR{ \exists t > 0 \text{ \, s.t. \,} |S_t| \leq \beta n + 1 \text{ and } |U_t| = 0 } &=
\PR{ \exists t> 0 \text{ \, s.t. \,} |S_t| \leq \beta n + 1 \text{ and }  |U_t| \leq 0 } \\
&\leq \PR{ \exists t> 0 \text{ \, s.t. \, }  \sum_{i=1}^t \inparen{ (1 +\alpha) X_i - 1  } \leq -\alpha}. %
\end{align*}

Let $Y_i = (\alpha +1 )X_i - 1$, so that $Y_i = \alpha$ with probability $p$, and $Y_i = -1$ with probability $1 - p$.  Let $Z_t = \sum_{i=1}^t Y_i$ be a random walk.  The above shows that
\begin{equation}\label{eq:randomwalk2}
 \PR{ \exists t > 0 , |S_t| \leq \beta n + 1 \text{ and } |U_t| = 0 } \leq \PR{ \exists t>0, Z_t = - \alpha}.
\end{equation}
Let $\tau_k(Y_i)$ denote the smallest $t$ so that $\inabs{\sum_{i=1}^t Y_i} \geq |k|$ and $\sgn(\sum_{i=1}^t Y_i) = \sgn(k)$, so that \eqref{eq:randomwalk2} reads: for all $M > 0$,
\begin{equation}\label{eq:randomwalk}
 \PR{ \exists t > 0, |S_t| \leq \beta n + 1 \text{ and } |U_t| = 0 } \leq 1 - \PR{ \tau_M(Y_i) \leq \tau_{-\alpha}(Y_i) }.
\end{equation}

Thus our goal is to show that the probability on the right hand side is bounded away from zero for any $M$; then there will be some constant probability that the edge-exploration process will make progress as long as $|S_t| \leq \beta n + 1$.

\begin{claim}
Using the notation above, for all $M > 0$,
\[ \PR{ \tau_M(Y_i) \leq \tau_{-\alpha}(Y_i) } \geq \frac{ \eps }{8}. \]
\end{claim}
\begin{proof}
We will use the following fact (see, for example \cite{GS12}, Chapter 12.2):
\begin{fact}[Asymmetric Gambler's Ruin]\label{fact:gambler}
Let $Y_i$ be independent random variables so that $Y_i = +1$ with probability $\gamma$ and $-1$ otherwise.
Let $a, b \geq 0$ be integers, and let $\vphi := (1-\gamma)/\gamma$.
Then
\[ \PR{ \tau_a(Y_i) \leq \tau_{-b}(Y_i) } = \begin{cases} \frac{ 1 - \vphi^b }{1 - \vphi^{b+a} } & \gamma \neq 1/2 \\
\frac{b}{b + a} & \gamma = 1/2 \end{cases} \]
\end{fact}

Using this fact, we will replace our steps $Y_t$ with steps $W_t \in \{\pm \Delta\}$, so that the behavior is the same (or worse for us).  
The analysis is made a bit hairier because $\alpha$ may not be an integer; in the following, the reader may wish to assume that $\alpha \in \mathbb{Z}$ and to take $\Delta = 1$.  

Let $\Delta > 0$ be some small constant which we will choose below, and let $a = \left\lfloor \frac{ \alpha }{ \Delta } \right\rfloor$ and $b = \left\lceil \frac{ 1 }{\Delta} \right\rceil$.  In particular, 
\begin{equation}\label{eq:realclose}
\alpha - \Delta \leq \Delta a \leq \alpha \qquad \text{and} \qquad 1 \leq \Delta b \leq 1 + \Delta.
\end{equation}
(Notice that if $\alpha$ is an integer then we may just take $\Delta = 1, \alpha = a, b = 1$ and we have $\Delta a = \alpha$ and $\Delta b = 1$).
Let $V_i$ be independent random variables so that $V_i = \Delta a$ with probability $p$ and $V_i = -\Delta b$ with probability $1 - p$.  Thus, as $\Delta \to 0$, we have $\Delta a \to \alpha$, $\Delta b \to 1$, and the behavior of $V_i$ approximates that of $Y_i$.

Thus, for any $M$, the probability that $\sum_{i=1}^t Y_i$ reaches $M$ before reaching $-\alpha$ is
\begin{align*}
\PR{ \tau_M(Y_i) \leq \tau_{-\alpha}(Y_i) } %
&\geq
\PR{ \tau_M(V_i) \leq \tau_{-\alpha}(V_i) } \\%
&\geq 
\PR{ \tau_{M}(V_i) \leq \tau_{-\Delta a}(V_i) }
\end{align*}
using the fact that $V_i$ is weakly more ``positive-going'' than $Y_i$ (in that $\Delta a \geq \alpha$ and $- \Delta b \geq -1$), and the last line follows since $-\Delta a \geq -\alpha$.

Now let $W_t$ be independent random variables so that $W_t = \Delta$ with probability $\gamma$ and $-\Delta$ with probability $1 - \gamma$, so that
$\vphi = (1- \gamma)/\gamma$ (which will be chosen below) satisfies
\begin{equation}\label{eq:want}
p \geq \frac{ 1 - \vphi^b }{ 1 - \vphi^{a + b }}. 
\end{equation}
Then by Fact~\ref{fact:gambler}, the probability that $\sum_{i=1}^t W_t$ reaches $\Delta a$ before reaching $-\Delta b$ is $p' \leq p$. Then we may couple the random walk with steps $W_t$ to the random walk with steps $V'_t$ which are $\Delta a$ with probability $p'$ and $-\Delta b$ otherwise.

Now we have
\[ \PR{ \tau_M( V_i ) \leq \tau_{-\Delta a}(V_i) } \geq \PR{ \tau_M(V_i') \leq \tau_{-\Delta a}(V_i') } = \PR{ \tau_M(W_i) \leq \tau_{-\Delta a}(W_i) }. \]
Using Fact~\ref{fact:gambler} again, we see that the right hand side above is equal to
\begin{equation}\label{eq:fail}
 \frac{ 1 - \vphi^{a} }{ 1 - \vphi^{M/\Delta + a} }. 
\end{equation}

Thus, we would like to show that for small enough $\Delta$, there is a choice of $\vphi$ so that \eqref{eq:want} is satisfied, while \eqref{eq:fail} is bounded away from zero.

To that end, choose
\[ \vphi = \exp\inparen{ - \eps/(4a) }. \]
We first check that \eqref{eq:want} holds, using the assumption that $p \geq (1 + \eps)/\alpha$. 
First, let us assert that $\Delta$ is small enough so that
\[ \frac{1 + \Delta}{\alpha - \Delta} \leq \frac{1}{\alpha}\inparen{ 1 + 3 \Delta} \leq \frac{1 + \eps}{1 + \frac{3\eps}{4} }. \]
(Above, this is possible since we are assuming that $\alpha \geq 1$).
This implies that (using \eqref{eq:realclose})
\begin{equation}\label{eq:bagood}
 \frac{b}{a} = \frac{\Delta b}{\Delta a} \leq \frac{1 + \Delta }{\alpha - \Delta} \leq \frac{1}{\alpha}(1 + 3\Delta) \leq \frac{ 1 + \eps }{1 + 3\eps/4 }
\end{equation}
which we will use below.  Now,
we have
\begin{align*}
\frac{ 1 - \vphi^b }{ 1 - \vphi^{a + b} } &= 
\frac{ 1 - \exp( -\eps b /(4a) ) }{1 - \exp( - \eps/4 - \eps b/(4a) ) }\\
&\leq \frac{ \eps b /(4a) }{ \frac{\eps}{4}( 1 + b/a ) - \frac{\eps^2}{32}(1 + b/a)^2 } \\
&\leq \frac{b}{a} \inparen{ \frac{1}{1 - \frac{\eps}{8}(1 + b/a) }} \\
&\leq \frac{b}{a} \inparen{ 1 + \frac{3\eps}{4} } 
\end{align*}
where above we have used
the fact that $1 - x \leq \exp(-x) \leq 1 - x + x^2/2$, and that $b/a < 2$ by \eqref{eq:bagood}, and the fact that $\eps \leq 1/3$ in the final line.
Continuing, we conclude
\begin{align*}
\frac{ 1 - \vphi^b }{ 1 - \vphi^{a + b} } 
&\leq \frac{b}{a} \inparen{ 1 + \frac{3\eps}{4} }  \\
&\leq \frac{1}{\alpha}(1 + 3\Delta )\inparen{ 1 + \frac{ 3\eps }{4} } \\
&\leq \frac{1 + \eps}{\alpha} \\
&\leq p
\end{align*}
where above we have again used \eqref{eq:bagood}. 

Next we check that \eqref{eq:fail} is at least $\eps/8$.  To do this, notice that
\begin{align*}
 \frac{ 1 - \vphi^{a} }{ 1 - \vphi^{M/\Delta + a} }
&\geq 1 - \vphi^a \\
&= 1 - \exp(-\eps/4) \\
&\geq \frac{\eps}{8},
\end{align*}
as desired, using the assumption that $\eps \leq 1/3$ in the last line.

Thus, putting everything together, we conclude that
\begin{align*}
\PR{ \tau_M(Y_i) \leq \tau_{-\alpha}(Y_i) }
&\geq
\PR{ \tau_{M}(V_i') \leq \tau_{-\Delta a}(V_i') } \\
&= \PR{ \tau_M(W_i) \leq \tau_{-\Delta a}(W_i) }\\
&=
 \frac{ 1 - \vphi^{a} }{ 1 - \vphi^{M/\Delta + a} } \\ 
&\geq \frac{\eps}{8}
\end{align*}
which proves the claim.
\end{proof}

Using the claim and \eqref{eq:randomwalk}, we have that
 \[ \PR{ \exists t > 0 \text{ s.t. } |S_t| \leq \beta n + 1 , |U_t| = 0 } \leq  1 - \eps/8. \]

Finally, 
\begin{align*}
\PR{ |C_{(u,v)}| \geq \beta n }
&= \PR{ |C_{(u,v)}| \geq \beta n \, \mid \, (u,v) \in E' } \cdot \PR{ (u,v) \in E' } \\
&= p \cdot \inparen{ 1 - \PR{ \exists t \text{ s.t. } |S_t| \leq \beta n + 1 , |U_t| = 0}}\\
&\geq \frac{p \eps }{8}. 
\end{align*}

This completes the proof.

\end{proof}

%% file: figure.tex
\begin{tikzpicture}
\tikzstyle{sedge}=[nicepink, ultra thick]
\tikzstyle{bedge}=[nicepeach, ultra thick,dashed]
\tikzstyle{uedge}=[nicepurple, ultra thick,densely dotted]
\tikzstyle{nedge}=[black]
\tikzstyle{nv} = [circle, black, draw, fill=black, scale=.3]
\tikzstyle{sv} = [circle, black, draw, fill=nicepink, scale=.5]

\begin{scope}
\node[sv] (t1) at (0,0) {};
\node[sv] (t2) at (1,0) {};
\draw[sedge] (t1) -- (t2);
\node at (-.5, 0) {$S_t$};
\end{scope}
\begin{scope}[yshift=-.5cm]
\node[sv] (t1) at (0,0) {};
\node[nv] (t2) at (1,0) {};
\draw[bedge] (t1) -- (t2);
\node at (-.5, 0) {$B_t$};
\end{scope}
\begin{scope}[yshift=-1cm]
\node[sv] (t1) at (0,0) {};
\node[nv] (t2) at (1,0) {};
\draw[uedge] (t1) -- (t2);
\node at (-.5, 0) {$U_t$};
\end{scope}
\begin{scope}[yshift=-1.5cm]
\node[sv] (t1) at (.5,0) {};
\node at (-.5, 0) {$N(S_t)$};
\end{scope}
\begin{scope}[xshift=6cm,yshift=-.75cm]
\node[nv](a) at (0,0) {};
\node[left=0cm of a] {$u$};
\node[nv](b) at (1,0) {};
\node[nv](c) at (1,1) {};
\node[above=0cm of c] {$v$};
\node[nv](d) at (1,-1) {};
\node[nv](e) at (2,1) {};
\node[nv](f) at (2,0) {};
\node[nv](g) at (2,-1) {};
\node[nv](h) at (3,0) {};
\draw[nedge] (a) -- (c);
\draw[nedge] (a) -- (b);
\draw[nedge] (a) -- (d);
\draw[nedge] (c) -- (f);
\draw[nedge] (d) -- (f);
\draw[nedge] (b) -- (f);
\draw[nedge] (b) -- (e);
\draw[nedge] (e) -- (f);
\draw[nedge] (g) -- (f);
\draw[nedge] (h) -- (e);
\draw[nedge] (h) -- (g);
\node at (-1,0) {$G$:};
\end{scope}

\begin{scope}[yshift=-4cm,xshift=-3cm]
\node at (1.5, -1.5) {$t=0$};
\node[sv](a) at (0,0) {};
\node[left=0cm of a] {$u$};
\node[nv](b) at (1,0) {};
\node[sv](c) at (1,1) {};
\node[above=0cm of c] {$v$};
\node[nv](d) at (1,-1) {};
\node[nv](e) at (2,1) {};
\node[nv](f) at (2,0) {};
\node[nv](g) at (2,-1) {};
\node[nv](h) at (3,0) {};
\draw[sedge] (a) -- (c);
\draw[uedge] (a) -- (b);
\draw[uedge] (a) -- (d);
\draw[uedge] (c) -- (f);
\draw[nedge] (d) -- (f);
\draw[nedge] (b) -- (e);
\draw[nedge] (b) -- (f);
\draw[nedge] (e) -- (f);
\draw[nedge] (g) -- (f);
\draw[nedge] (h) -- (e);
\draw[nedge] (h) -- (g);
\end{scope}

\begin{scope}[yshift=-4cm,xshift=1cm]
\node at (1.5, -1.5) {$t=1$};
\node[sv](a) at (0,0) {};
\node[left=0cm of a] {$u$};
\node[sv](b) at (1,0) {};
\node[sv](c) at (1,1) {};
\node[above=0cm of c] {$v$};
\node[nv](d) at (1,-1) {};
\node[nv](e) at (2,1) {};
\node[nv](f) at (2,0) {};
\node[nv](g) at (2,-1) {};
\node[nv](h) at (3,0) {};
\draw[sedge] (a) -- (c);
\draw[sedge] (a) -- (b);
\draw[uedge] (a) -- (d);
\draw[uedge] (c) -- (f);
\draw[nedge] (d) -- (f);
\draw[uedge] (b) -- (e);
\draw[uedge] (b) -- (f);
\draw[nedge] (e) -- (f);
\draw[nedge] (g) -- (f);
\draw[nedge] (h) -- (e);
\draw[nedge] (h) -- (g);
\end{scope}

\begin{scope}[yshift=-4cm,xshift=5cm]
\node at (1.5, -1.5) {$t=2$};
\node[sv](a) at (0,0) {};
\node[left=0cm of a] {$u$};
\node[sv](b) at (1,0) {};
\node[sv](c) at (1,1) {};
\node[above=0cm of c] {$v$};
\node[nv](d) at (1,-1) {};
\node[nv](e) at (2,1) {};
\node[nv](f) at (2,0) {};
\node[nv](g) at (2,-1) {};
\node[nv](h) at (3,0) {};
\draw[sedge] (a) -- (c);
\draw[sedge] (a) -- (b);
\draw[uedge] (a) -- (d);
\draw[uedge] (c) -- (f);
\draw[nedge] (d) -- (f);
\draw[bedge] (b) -- (e);
\draw[uedge] (b) -- (f);
\draw[nedge] (e) -- (f);
\draw[nedge] (g) -- (f);
\draw[nedge] (h) -- (e);
\draw[nedge] (h) -- (g);
\end{scope}
\begin{scope}[yshift=-4cm,xshift=9cm]
\node at (1.5, -1.5) {$t=3$};
\node[sv](a) at (0,0) {};
\node[left=0cm of a] {$u$};
\node[sv](b) at (1,0) {};
\node[sv](c) at (1,1) {};
\node[above=0cm of c] {$v$};
\node[nv](d) at (1,-1) {};
\node[nv](e) at (2,1) {};
\node[sv](f) at (2,0) {};
\node[nv](g) at (2,-1) {};
\node[nv](h) at (3,0) {};
\draw[sedge] (a) -- (c);
\draw[sedge] (a) -- (b);
\draw[uedge] (a) -- (d);
\draw[nedge] (c) -- (f);
\draw[nedge] (d) -- (f);
\draw[bedge] (b) -- (e);
\draw[sedge] (b) -- (f);
\draw[uedge] (e) -- (f);
\draw[uedge] (g) -- (f);
\draw[nedge] (h) -- (e);
\draw[nedge] (h) -- (g);
\end{scope}

\end{tikzpicture}

%% file: empirical.tex
In this section, we numerically compare Algorithm~\ref{alg:make-tests} to existing work. We compare to the random walk based approach of \cite{Cheraghchi:2010vm} and to the randomized group testing \em without \em graph constraints of Proposition~\ref{prop:gt-upper}.
We find that the random subgraph tests perform nearly as well as the unconstrained versions in most settings, and often perform significantly better than the random walk approach.

Below, we test the following randomized constructions of tests:
\begin{itemize}
\item Our approach, Algorithm~\ref{alg:make-tests} (called ``Subgraph'' in the figures).  We use $p=1/(d+1)$ and include only the largest connected component of $G(p)$. 
\item The random walk approach of \cite{Cheraghchi:2010vm} (called ``Random walks'' in the figures).  We empirically estimate the mixing time $\tau$ by picking a node at random and finding the first time that the total variation distance to the equilibrium distribution is less than $1/(2cn)^2$, as per the definition in \cite{Cheraghchi:2010vm}, where $c$ is defined so that the graph $G=(V,E)$ has $D \leq \deg(v) \leq cD$ for each $v \in V$.  We fix a constant $l > 0$ and run each random walk for $\ceil{\frac{lnD}{c^3d\tau}}$ steps. We tried a few of values of $l$ and present the best results for each graph: for the complete graph we chose $l=1$ and for all other graphs we chose $l=4$.
\item Unconstrained random approach (called ``Random'' in the figures).  We include each edge in a test with probability $p = 1/(d+1)$, and ignore the graph constraints. 
\end{itemize}
We consider four types of graphs.  The first three---random regular graphs, complete graphs, and hypercubes---are idealized graphs that may or may not capture real networks.  For our last graph, we choose the ``Fat-Tree'' graph~\cite{L85}, originally designed for use in supercomputers and which is now widely used in datacenter networks~\cite{AlFares:2008tv}.  As the name suggests, this is a ``fattened'' tree, where the fatness (number of links) near the top of the tree is greater than the fatness near the leaves. (See Figure~\ref{fig:fattree}).

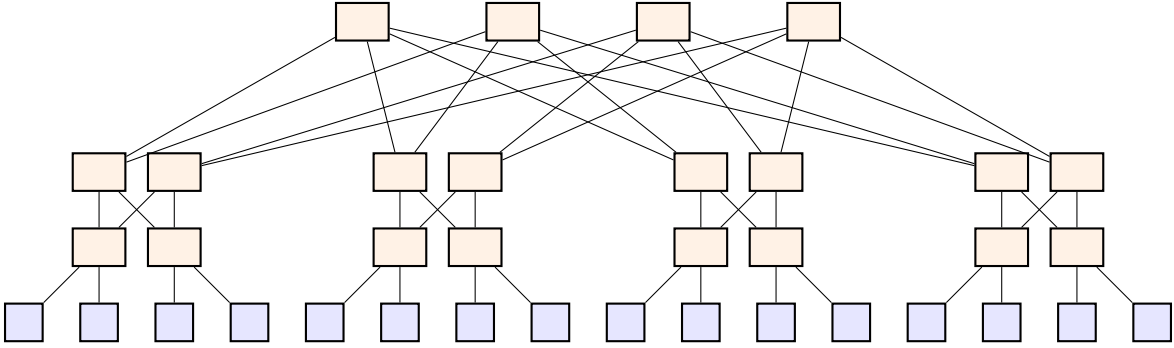
\begin{figure}
\begin{center}
\begin{tikzpicture}
\tikzstyle{client}=[draw,thick, fill=blue!10, rectangle, minimum height=.5cm, minimum width=.5cm]
\tikzstyle{switch}=[draw,thick, fill=orange!10, rectangle, minimum height=.5cm, minimum width=.7cm]

\node[switch](A) at (-3,0) {};
\node[switch](B) at (-1, 0) {};
\node[switch](C) at (1, 0) {};
\node[switch](D) at (3,0) {};

\foreach \shift in {-7.5, -3.5, .5, 4.5}{
\begin{scope}[xshift=\shift cm, yshift=-4cm]
\node[client](a1) at (0,0) {};
\node[client](b1) at (1,0) {};
\node[client](c1) at (2,0) {};
\node[client](d1) at (3,0) {};

\node[switch](x1) at (1,1) {};
\node[switch](z1) at (2,1) {};
\node[switch](y1) at (1,2) {};
\node[switch](w1) at (2,2) {};

\draw (a1) to (x1);
\draw (b1) to (x1);
\draw (c1) to (z1);
\draw (d1) to (z1);

\draw (x1) to (y1);
\draw (x1) to (w1);
\draw (z1) to (w1);
\draw (z1) to (y1);

\draw (y1) to (A);
\draw (y1) to (B);
\draw (w1) to (C);
\draw (w1) to (D);
\end{scope}
}
\end{tikzpicture}
\end{center}
\caption{A small example of the fat tree topology with $n=36,m=48$.  In our experiments, we consider larger versions ($n=80,m=256$ and $n=45,m=108$).}\label{fig:fattree}
\end{figure}

We perform two types of experiments:
\begin{enumerate}
\item   In the first type of experiment, we compare the probability of obtaining a $d$-disjunct matrix from any of these three randomized approaches.  
Unfortunately, it is computationally intense to determine whether or not a given collection $\mathcal{T}$ of tests is $d$-disjunct, and so we are only able to do this for small $d$ ($d=1$ and $d=2$). 
\item In the second type of experiment, we are trying to understand the performance of our method for larger $d$.
Since determining $d$-disjunctness is computationally infeasible for large $d$, instead
we choose $d$ random defectives and estimate the probability of success under each of the three methods.  As mentioned in Remark~\ref{remark:random}, we believe that our theoretical approach should also work for random defectives, although the primary goal of these experiments is to get an idea of how our approach might work in practice.
\end{enumerate}

\subsection{Tests for $d$-disjunctness} 

First, we estimate the number of tests required for $d$-disjunctness for Algorithm~\ref{alg:make-tests} and compare it to \cite{Cheraghchi:2010vm} and randomized group testing without graph constraints. We find that for many graphs, our approach requires roughly the same number of tests as group testing without constraints.

Figure~\ref{fig:p-1-disjunct} (resp. Figure~\ref{fig:p-2-disjunct}) shows the probability that a randomly generated test matrix is 1-disjunct (resp. 2-disjunct) for various graphs, algorithms, and numbers of tests. Each point is the empirical mean of $100$ independent trials, and we plot error bars of width $0.1 = 1/\sqrt{100}$.  (Notice that by Hoeffding's inequality, the probability that the true average lies outside the error bars is at most $1/e^2$).

Algorithm~\ref{alg:make-tests} performs similarly to the nearly optimal randomized group testing procedure of Proposition~\ref{prop:gt-upper}. Notably, for the Fat-Tree graph, Algorithm~\ref{alg:make-tests} significantly outperforms the approach of \cite{Cheraghchi:2010vm}.

\begin{figure}[H]
  \centering

  \newcommand{\empiricalgraph}[2]{%
    \begin{subfigure}[b]{0.49\textwidth}%
      \includegraphics[width=\textwidth]{figures/#1}%
      \caption*{#2}%
    \end{subfigure}}

  \empiricalgraph{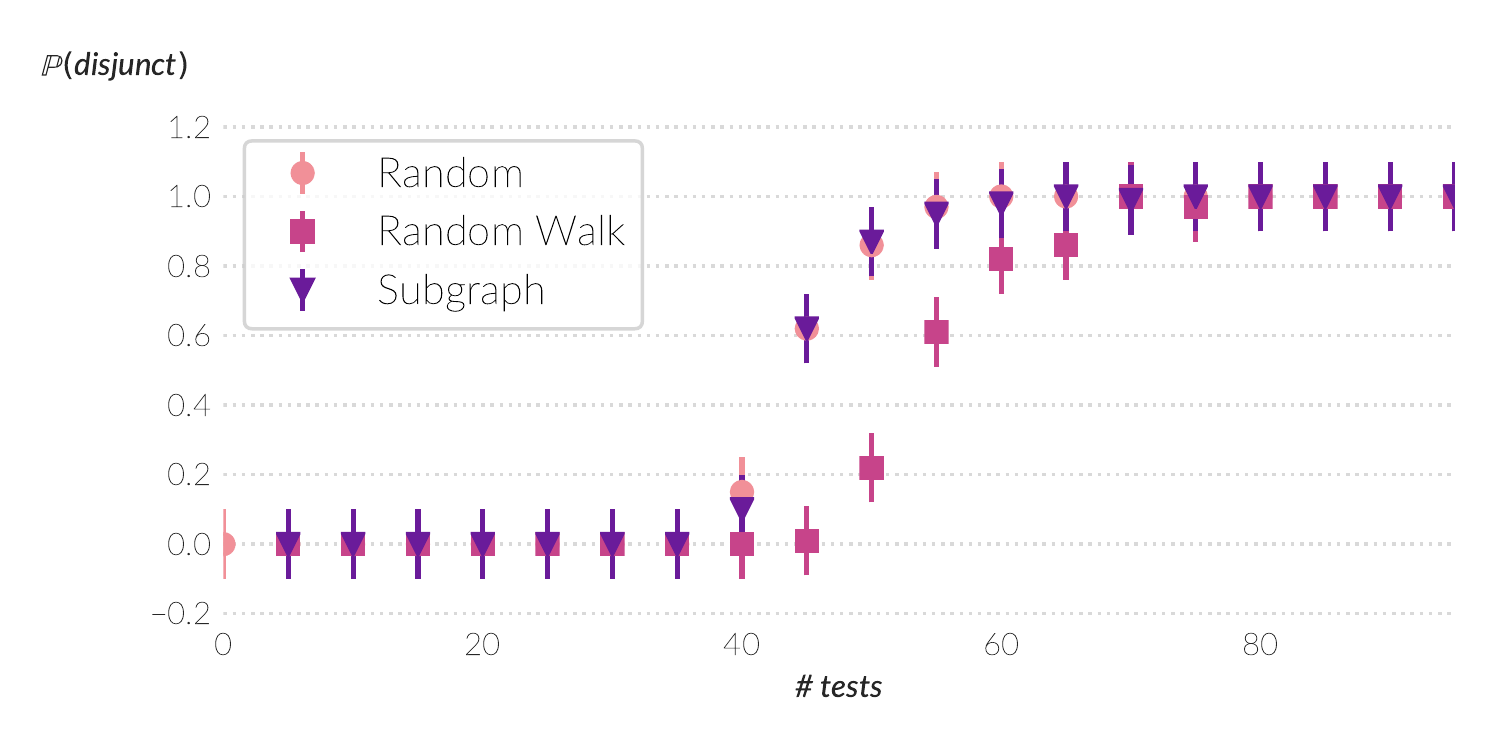}{Random Regular Graph, $n=100, m=500$}
  \empiricalgraph{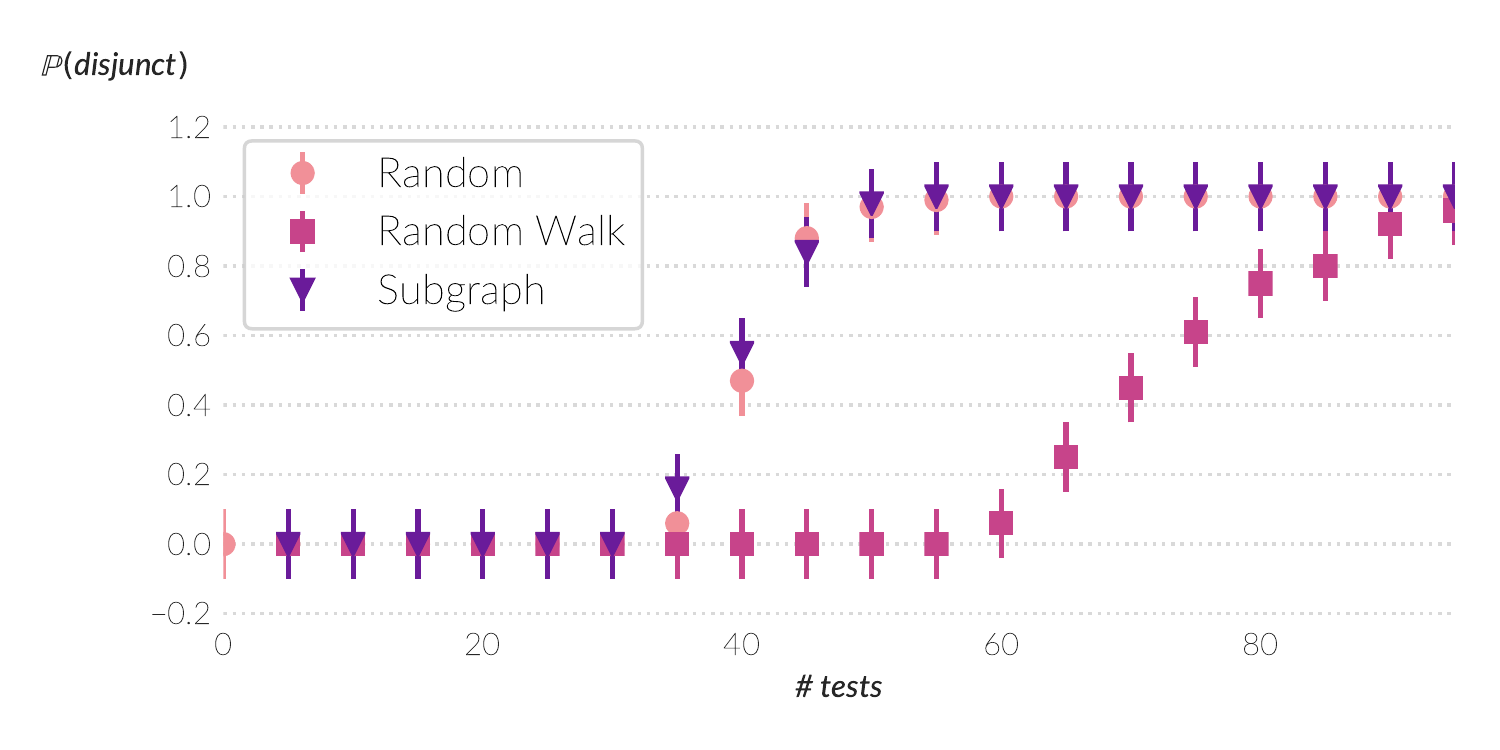}{Complete Graph, $n=23, m=253$}
  \empiricalgraph{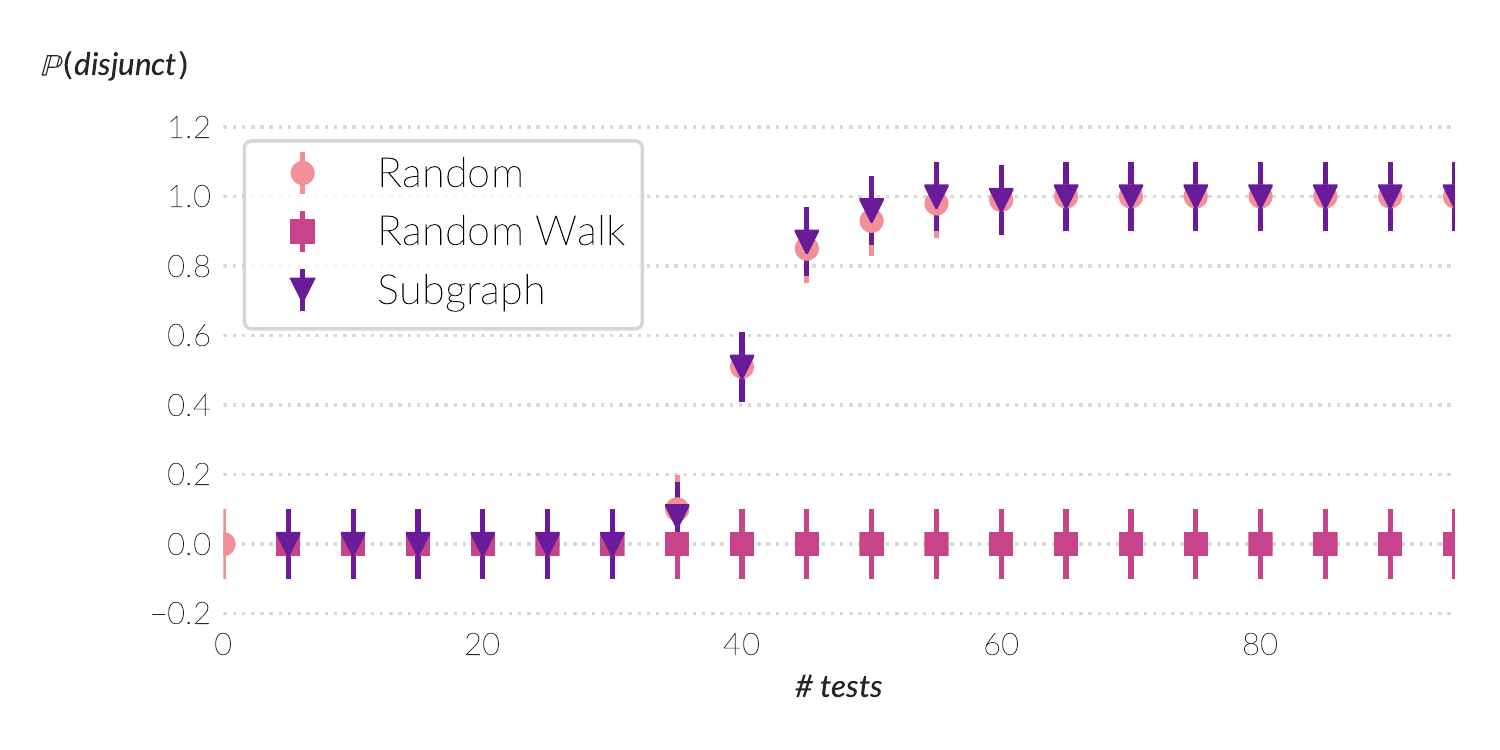}{Fat Tree, $n=80, m=256$}
  \empiricalgraph{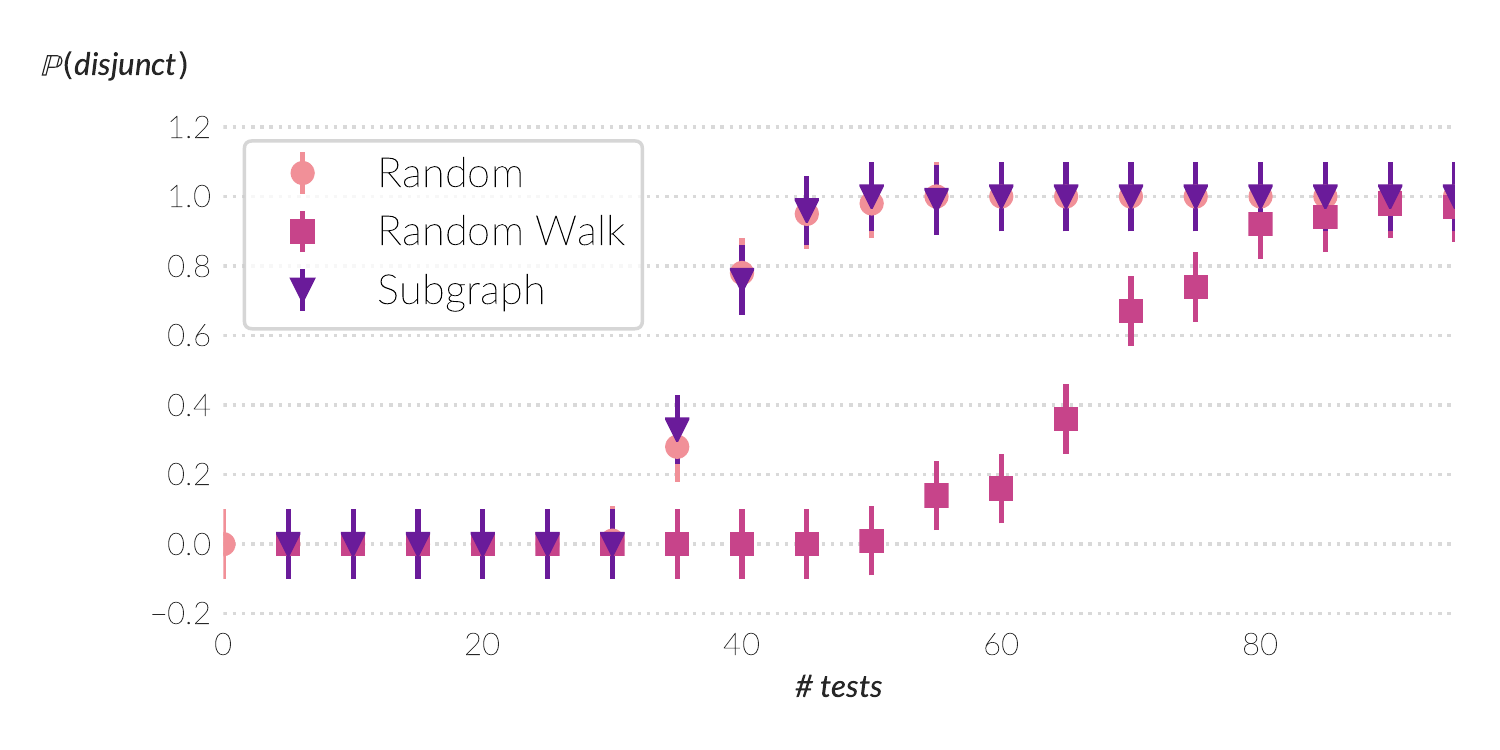}{Hypercube, $n=64, m=192$}
  \caption{Probability that a randomly generated test matrix with a certain number of tests is 1-disjunct for various graphs. Each point is the mean of 100 trials, and error bars of 0.1 are plotted. ``Subgraph'' is our approach, ``Random Walk'' the approach of \cite{Cheraghchi:2010vm}, ``Random'' the nearly-optimal randomized construction for unconstrained group testing.}
  \label{fig:p-1-disjunct}
\end{figure}

\begin{figure}[H]
  \centering

  \newcommand{\empiricalgraph}[2]{%
    \begin{subfigure}[b]{0.49\textwidth}%
      \includegraphics[width=\textwidth]{figures/#1}%
      \caption*{#2}%
    \end{subfigure}}

  \empiricalgraph{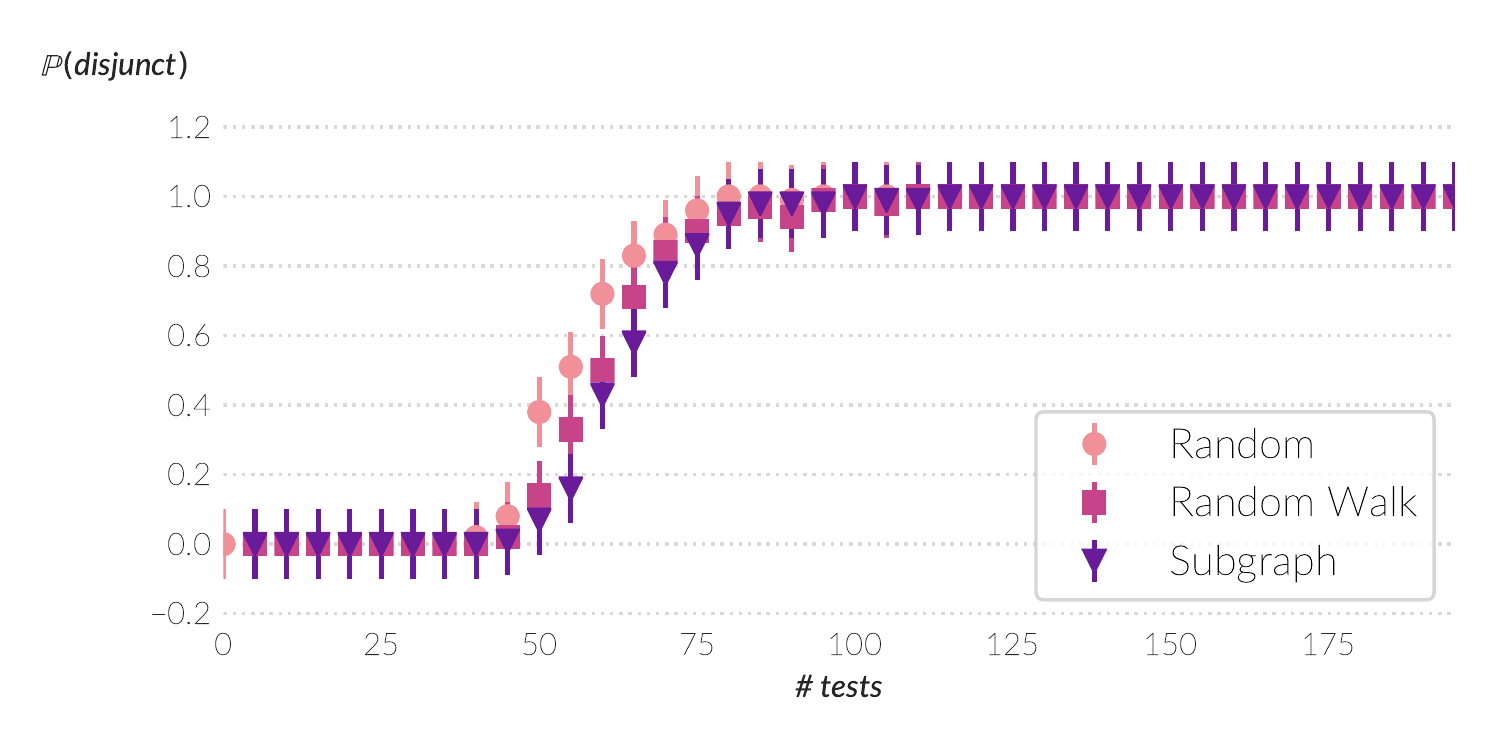}{Random Regular Graph, $n=10, m=25$}
  \empiricalgraph{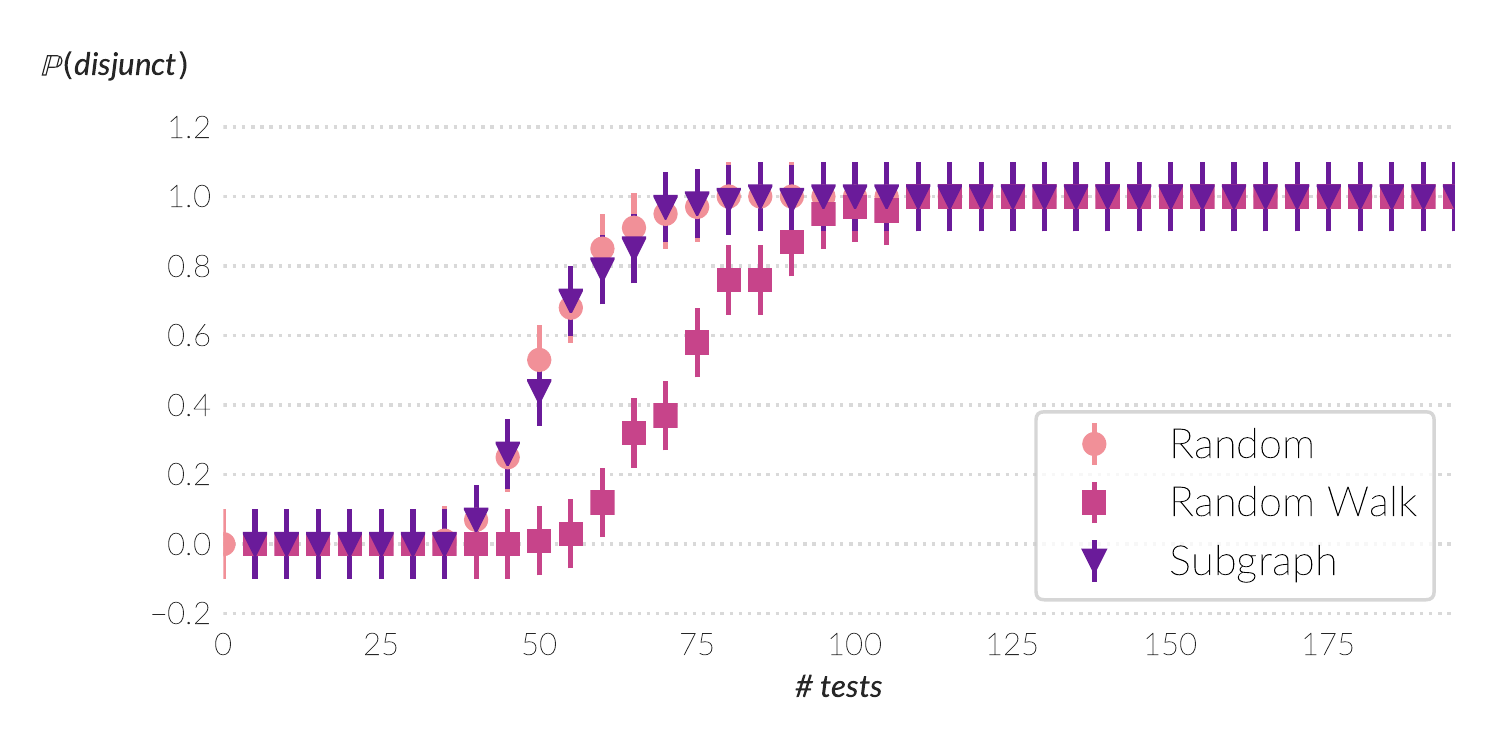}{Complete Graph, $n=7, m=21$}
  \empiricalgraph{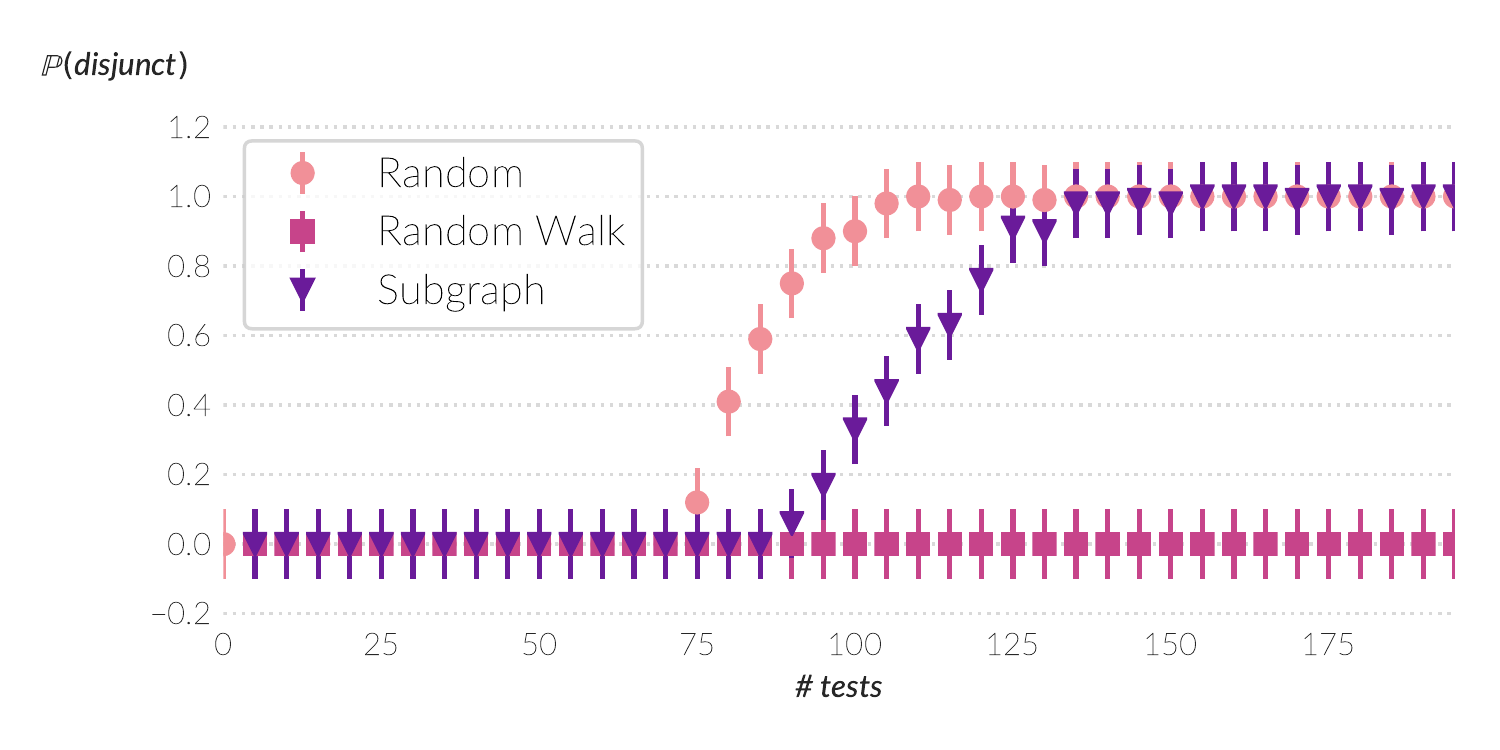}{Fat Tree, $n=45, m=108$}
  \empiricalgraph{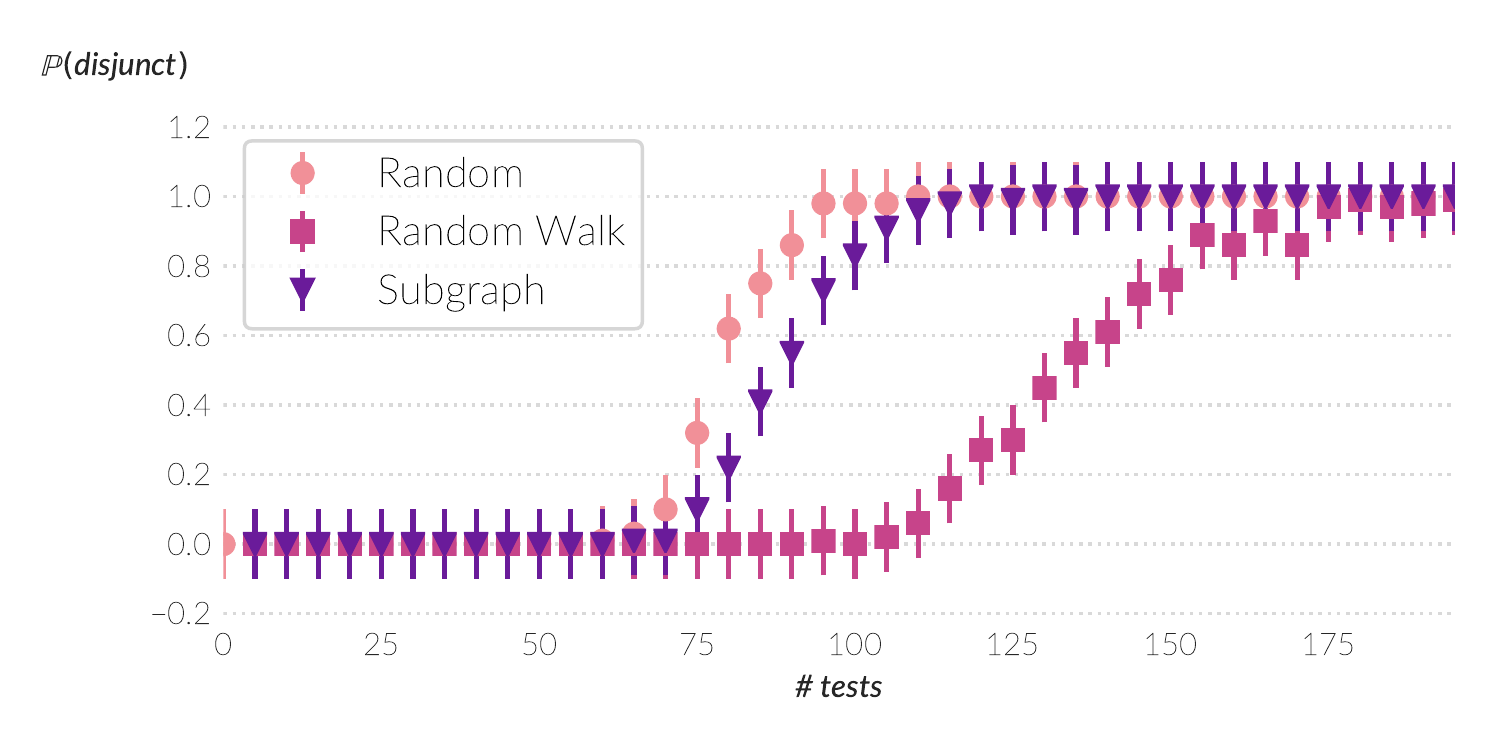}{Hypercube, $n=32, m=80$}
  \caption{Probability that a randomly generated test matrix with a certain number of tests is 2-disjunct for various graphs. Each point is the mean of 100 trials, and error bars of 0.1 are plotted. ``Subgraph'' is our approach, ``Random Walk'' the approach of \cite{Cheraghchi:2010vm}, ``Random'' the nearly-optimal randomized construction for unconstrained group-testing.}
  \label{fig:p-2-disjunct}
\end{figure}

\subsection{Tests on $d$ random failures}

As mentioned above, determining $d$-disjunctness is computationally infeasible for larger $d$, and so to assess larger $d$ we consider performance on random failures.  This model has been considered both in group testing (e.g. \cite{M16,LPR16,HKWO18}) and boolean network tomography \cite{Duffield:2006ue,Nguyen:2007ds}.  It is not hard to see that without graph constraints, including each edge with probability $p = 1/(d+1)$ will, with high probability over both the tests and the failures, identify up to $d$ random failures with $O(d\log(m))$ tests.  As noted in Remark~\ref{remark:random}, we believe our approach can provably achieve similar results.

For our experiments with random failures, we focus on the fat-tree topology.  The main reasons for this are (a) that the ``Fat-Tree with random failures'' set-up is perhaps the most relevant for real-life applications, and (b) the other topologies yield graphs that look similar, but the differences between the three approaches are less pronounced.  

We find that Algorithm~\ref{alg:make-tests} significantly out-performs the random walk approach of \cite{Cheraghchi:2010vm}, but performs less well as $d$ grows. In this graph once $d$ becomes much larger than 5, even the random group testing construction without graph constraints requires at least $m$ tests.

Figure~\ref{fig:p-correct} shows the probability that a set of tests correctly identifies $d$ random failures, where the probability is taken over both the tests and the failures.  
Each point is the empirical mean of 100 independent trials, and we plot error bars of width $0.1 = 1/\sqrt{100}$.  As above, by a Hoeffding bound the probability that the true mean lies outside the error bars is at most $e^{-2}$.

\begin{figure}[H]
  \centering

  \newcommand{\empiricalgraph}[2]{%
    \begin{subfigure}[b]{0.49\textwidth}%
      \includegraphics[width=\textwidth]{figures/#1}%
      \caption*{#2}%
    \end{subfigure}}

  \empiricalgraph{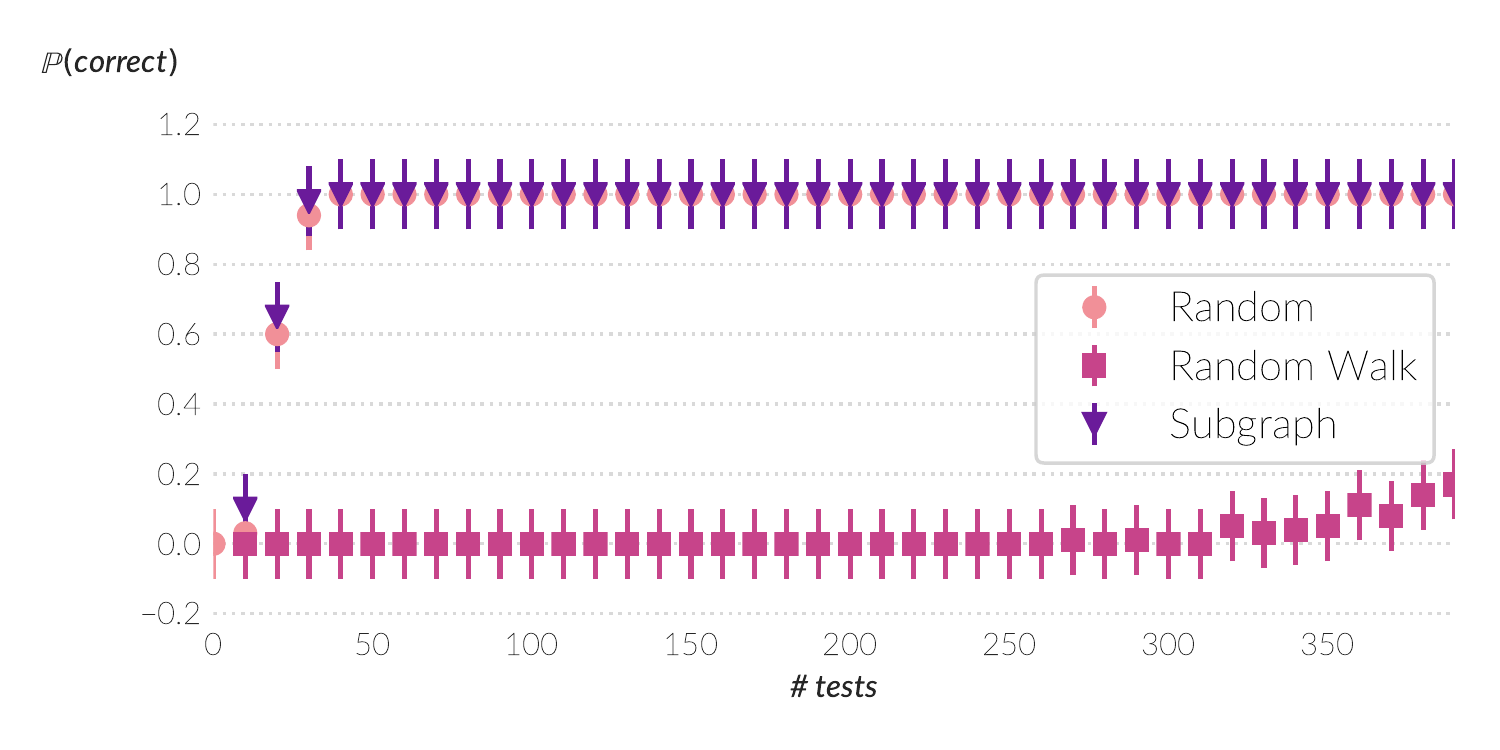}{Fat Tree, $n=80, m=256, d=1$}
  \empiricalgraph{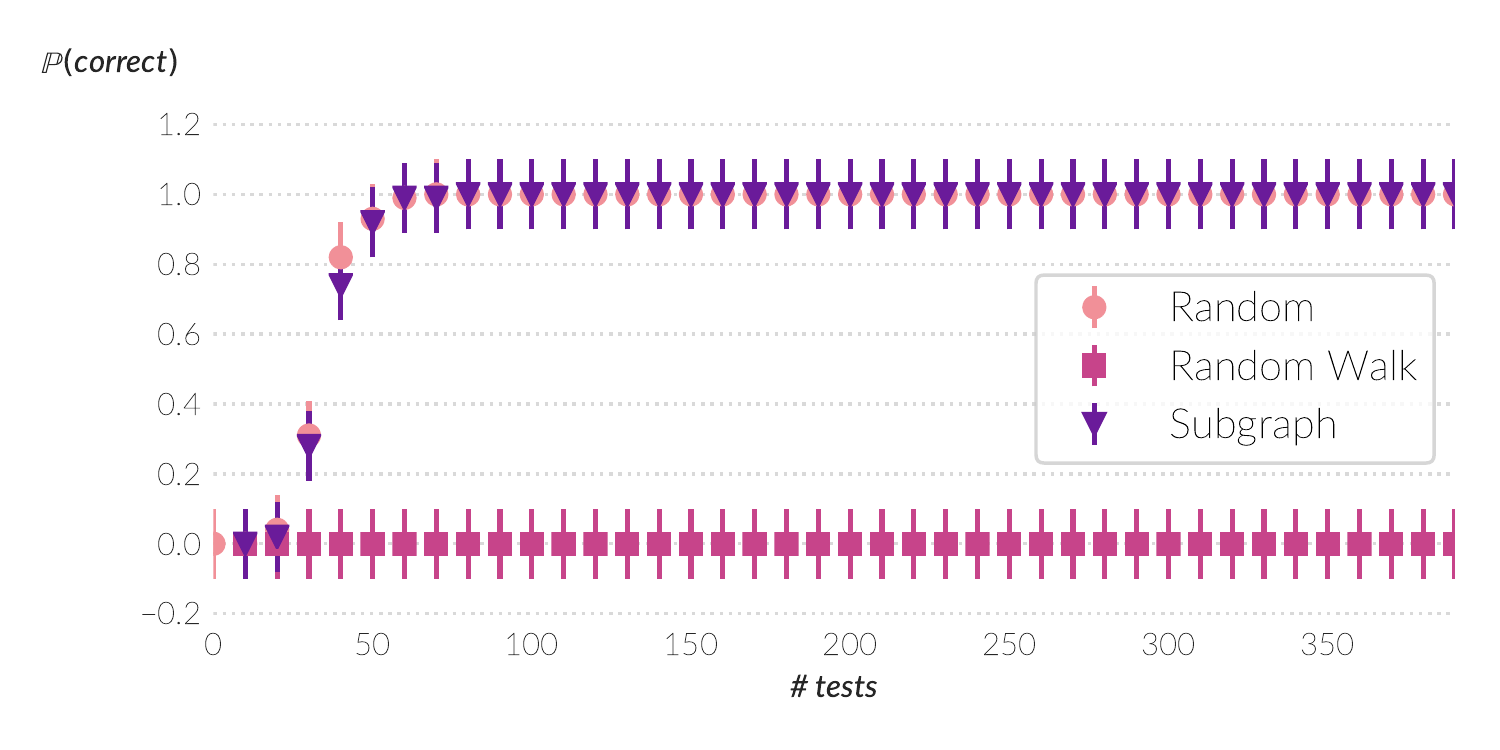}{Fat Tree, $n=80, m=256, d=2$}
  \empiricalgraph{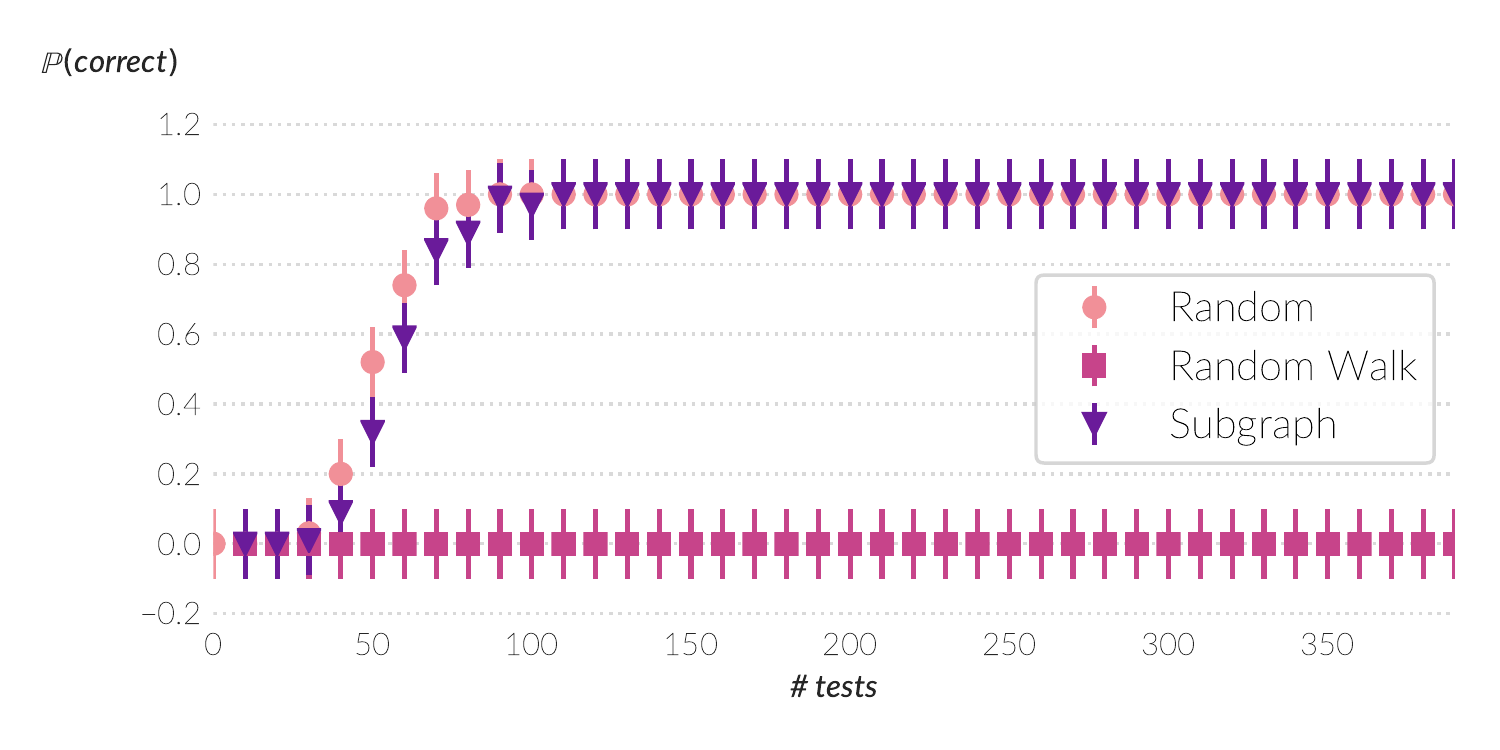}{Fat Tree, $n=80, m=256, d=3$}
  \empiricalgraph{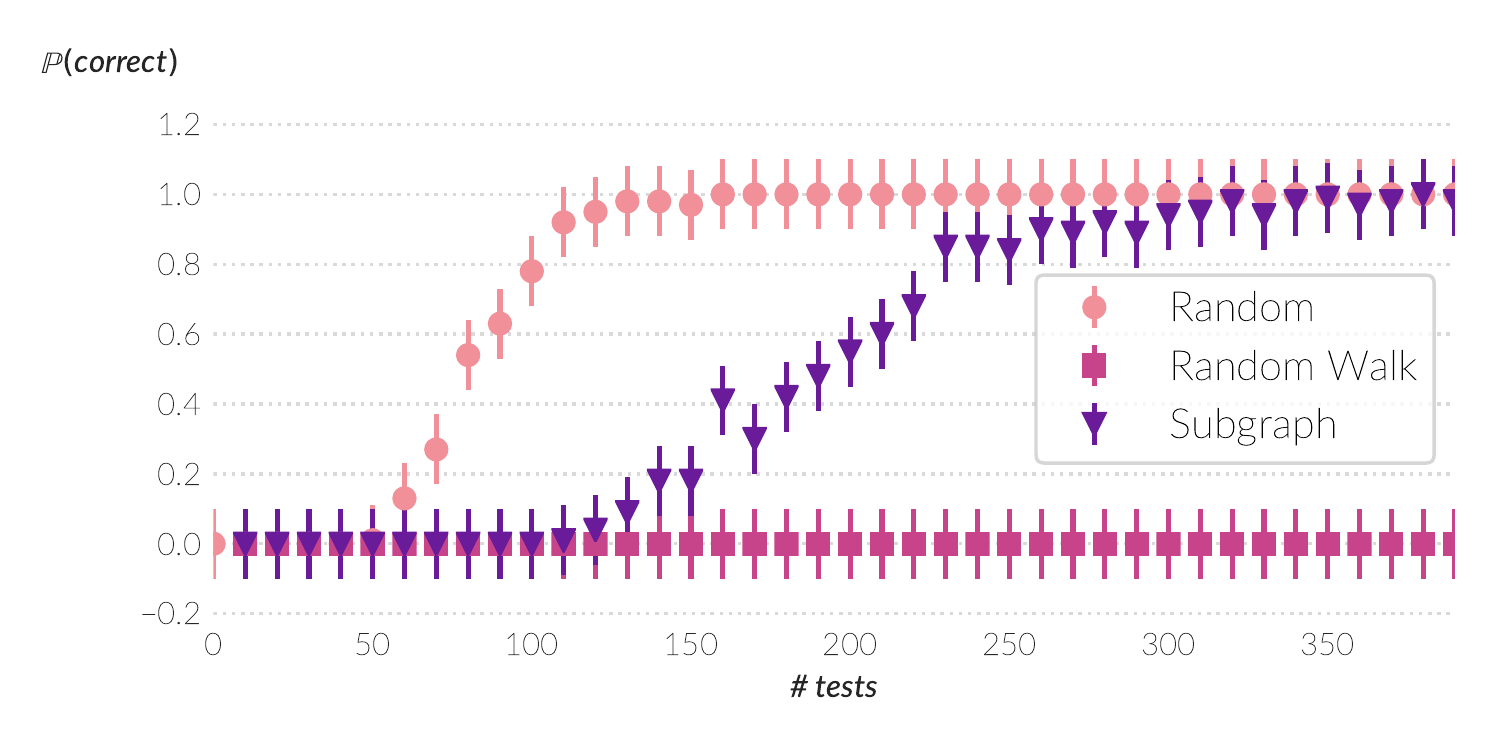}{Fat Tree, $n=80, m=256, d=5$}
  \caption{The probability that tests generated from various schemes  with a certain number of tests correctly identifies $d$ random failed edges. ``Subgraph'' is our approach, ``Random Walk'' the approach of \cite{Cheraghchi:2010vm}, ``Random'' the nearly-optimal randomized construction for unconstrained group testing.}
  \label{fig:p-correct}
\end{figure}

%% file: conclusion.tex
We have given a simple randomized construction which shows that for many graphs, graph-constrained group testing is possible with a near-optimal number of tests.  Our results---which are proved by analyzing a particular random walk---improve over previous work, and also apply to a wider range of graphs.  However, many open questions remain, and we conclude with a few of these here.

\begin{enumerate}
\item Both our approach and the approach of \cite{Cheraghchi:2010vm} give randomized constructions.  Derandomizing these constructions remains a fascinating open question.  Such a derandomization would be especially useful if it allowed a node to extremely efficiently determine which of its neighbors a test packet should be sent to next, using only minimal information stored in the packet.
\item While $(\beta, \alpha)$-expansion is reasonably general, it is not completely general.  For example, hypercubes are not very good $(\beta, \alpha)$-expanders, but the result of \cite{Harvey:2007ez} implies that (since they have many disjoint spanning trees) hypercubes are reasonably good for the graph-constrained group-testing problem: $O(d^3 \log_d(m))$ tests suffice to identify $d$ defectives for $d \lesssim \log(n)$.
It seems possible that one could modify our analysis using the approach of \cite{Ajtai:1982jx}---which shows that random sparsifications of hypercubes have large connected components with high probability---to obtain a good result for hypercubes as well.
Thus, it is an open question to see how well our approach works for hypercubes, but more generally if there is some quantity (more general than $(\beta, \alpha)$-edge expansion) which precisely captures when our approach works and when it does not.
\item In Appendix~\ref{app:separation}, we show that simple paths are not as powerful as connected-subgraph tests for graph-constrained group testing.  However, in practice it is often the case that simple paths (and especially shortest paths) are easier to implement.  It would be interesting to characterize the limitations of graph-constrained group testing when the tests are restricted to (shortest) simple paths.
\end{enumerate}

%% file: separation.tex
In this appendix, we briefly observe that connected subgraph tests are strictly more powerful for graph-constrained group testing than tests which are constrained either to be simple paths or trees.

First we observe a separation between simple-path tests and tree tests.  
Let $G$ be a balanced binary tree on $n$ nodes, and let $d=1$.
It is not hard to see that any collection $\mathcal{T}$ of simple path tests which are $1$-disjunct must have $|\mathcal{T}| = \Omega(n)$, because that many tests are required simply to cover $G$.  
However, Theorem 7 of \cite{Harvey:2007ez} shows that $O(\log^2(n))$ tree tests suffice to solve the graph-constrained group testing problem. 

Next we observe a separation between tree tests and connected-subgraph tests.
Let $G$ be the complete graph on $n$ vertices.  Any $d$-disjunct collection $\mathcal{T}$ of tree tests must have size $\Omega(n)$, again because $\Omega(n)$ trees are required to cover all the edges in the complete graph.
However \cite{Cheraghchi:2010vm} show that $O(d^2 \log(n/d))$ connected-subgraph tests suffice for the complete graph.